\documentclass[preprintnumbers,pra,
superscriptaddress,
longbibliography,
twocolumn
]{revtex4-2}
\usepackage[latin1]{inputenc}
\usepackage{graphicx}
\newcommand*{\medcap}{\mathbin{\scalebox{1.5}{\ensuremath{\cap}}}}%
\usepackage{amsmath}
\usepackage{amsthm}
\usepackage{bm}
\usepackage{layout}
\usepackage{float}
\usepackage{amsfonts}
\usepackage{amssymb}%
\usepackage[margin=0.75in]{geometry}
\usepackage{color}
\usepackage{soul}
\usepackage{mathtools}
\usepackage[colorlinks=true,citecolor=blue]{hyperref}
\theoremstyle{definition}
\newtheorem{theorem}{Theorem}

\newtheorem{lemma}{Lemma}

\newtheorem{corollary}{Corollary}
\usepackage[capitalise,compress]{cleveref}

\newcommand{\ket}[1]{\left | #1 \right\rangle}

\newcommand{\bra}[1]{\left \langle #1 \right |}

\newcommand{\abs}[1]{\left | #1 \right|}
\renewcommand{\epsilon}{\varepsilon}
\renewcommand{\O}[1]{ O\left(#1\right)}

\newcommand{\norm}[1]{\left\|#1\right\|}

\newcommand{\comm}[1]{\left[#1\right]}

\newcounter{para}

\makeatletter
\newcommand*\bigcdot{\mathpalette\bigcdot@{.5}}
\newcommand*\bigcdot@[2]{\mathbin{\vcenter{\hbox{\scalebox{#2}{$\m@th#1\bullet$}}}}}

\makeatother

\newcommand{\ad}{\text{ad}}

\usepackage{array}
\newcolumntype{L}{>{$}l<{$}} 
\newcolumntype{C}{>{$}c<{$}} 
\newcolumntype{R}{>{$}r<{$}} 

\usepackage{xcolor}

\usepackage{mathtools}

\usepackage{xr}
\makeatletter
\newcommand*{\addFileDependency}[1]{
  \typeout{(#1)}
  \@addtofilelist{#1}
  \IfFileExists{#1}{}{\typeout{No file #1.}}
}
\makeatother

\usepackage{tikz}

\newcommand{\poly}{\text{poly}}
\newcommand{\Texp}[1]{\mathcal T\exp\left\{#1\right\}}

\usepackage{mdframed}
\newmdenv[topline=false,rightline=false,bottomline=false,linewidth=2pt,linecolor=white!60!black,]{leftborder}

\newcommand{\dt}{{\delta t}}

\newcommand{\kick}{{\text{kick}}}
\newcommand{\Zeno}{{\text{Zeno}}}
\usepackage{grffile}

\usepackage[english]{babel}

\makeatletter
\def\bbl@set@language#1{%
  \edef\languagename{%
    \ifnum\escapechar=\expandafter`\string#1\@empty
    \else\string#1\@empty\fi}%
  \@ifundefined{babel@language@alias@\languagename}{}{%
    \edef\languagename{\@nameuse{babel@language@alias@\languagename}}%
  }%
  \select@language{\languagename}%
  \expandafter\ifx\csname date\languagename\endcsname\relax\else
    \if@filesw
      \protected@write\@auxout{}{\string\select@language{\languagename}}%
      \bbl@for\bbl@tempa\BabelContentsFiles{%
        \addtocontents{\bbl@tempa}{\xstring\select@language{\languagename}}}%
      \bbl@usehooks{write}{}%
    \fi
  \fi}
\newcommand{\DeclareLanguageAlias}[2]{%
  \global\@namedef{babel@language@alias@#1}{#2}%
}
\makeatother

\DeclareLanguageAlias{en}{english}

\newcommand{\normcomm}[1]{\norm{\comm{#1}}}

\newcommand{\overbar}[1]{\mkern 3mu\overline{\mkern-3mu#1\mkern-1.5mu}\mkern 1.5mu}
\newcommand{\eff}{\text{eff}}

\usepackage[caption=false]{subfig}

\begin{document}

\title{Faster Digital Quantum Simulation by Symmetry Protection}
\date{\today}

\author{Minh~C.~Tran}
\affiliation{Joint Center for Quantum Information and Computer Science, NIST/University of Maryland, College Park, Maryland 20742, USA}
\affiliation{Joint Quantum Institute, NIST/University of Maryland, College Park, Maryland 20742, USA}
\author{Yuan~Su}
\affiliation{Joint Center for Quantum Information and Computer Science, NIST/University of Maryland, College Park, Maryland 20742, USA}
\affiliation{Department of Computer Science, University of Maryland, College Park, Maryland 20742, USA}
\affiliation{Institute for Advanced Computer Studies, University of Maryland, College Park, Maryland 20742, USA}
\affiliation{Institute for Quantum Information and Matter, California Institute of Technology, Pasadena, California 91125, USA}
\author{Daniel Carney}
\affiliation{Joint Center for Quantum Information and Computer Science, NIST/University of Maryland, College Park, Maryland 20742, USA}
\affiliation{Joint Quantum Institute, NIST/University of Maryland, College Park, Maryland 20742, USA}
\affiliation{Fermi National Accelerator Laboratory, Batavia, Illinois 60510, USA}\
\author{Jacob~M.~Taylor}
\affiliation{Joint Center for Quantum Information and Computer Science, NIST/University of Maryland, College Park, Maryland 20742, USA}
\affiliation{Joint Quantum Institute, NIST/University of Maryland, College Park, Maryland 20742, USA}
\begin{abstract}
Simulating the dynamics of quantum systems is an important application of quantum computers and has seen a variety of implementations on current hardware.
We show that by introducing quantum gates implementing unitary transformations generated by the symmetries of the system, one can induce destructive interference between the errors from different steps of the simulation, effectively giving faster quantum simulation by symmetry protection.
We derive rigorous bounds on the error of a symmetry-protected simulation algorithm and identify conditions for optimal symmetry protection. 
In particular, when the symmetry transformations are chosen as powers of a unitary, the error of the algorithm is approximately projected to the so-called quantum Zeno subspaces. 
We prove a bound on this approximation error, exponentially improving a recent result of Burgarth, Facchi, Gramegna, and Pascazio.
We apply the symmetry protection technique to the simulations of the XXZ Heisenberg interactions with local disorder and the Schwinger model in quantum field theory. 
For both systems, the technique can reduce the simulation error by several orders of magnitude over the unprotected simulation.
Finally, we provide numerical evidence suggesting that the technique can also protect simulation against other types of coherent, temporally correlated errors, such as the $1/f$ noise commonly found in solid-state experiments.
\end{abstract}

\preprint{FERMILAB-PUB-20-240-QIS-T}

\maketitle

\section{Introduction}

Simulating the dynamics of quantum systems is a key application of quantum computers.
However, digitalizing the continuous time evolutions to enable execution on gate-based and other programmable quantum computers comes with simulation errors that cause the dynamics of the  systems to deviate from ideal evolutions.
In particular, the errors may violate the symmetries in the target Hamiltonian for simulation, resulting in unphysical states at the end of the simulations.
This digitalization error particularly affects Trotterization---the most common algorithm for near-term quantum simulations~\cite{blattQuantumSimulationsTrapped2012,monroeProgrammableQuantumSimulations2019,kjaergaardSuperconductingQubitsCurrent2020}---and persists even in more sophisticated, advanced quantum simulation algorithms~\cite{suzukiGeneralTheoryFractal1991,berrySimulatingHamiltonianDynamics2015c,lowHamiltonianSimulationQubitization2019b}. 

In this paper, we propose an approach, using the symmetries of a target Hamiltonian, to protect its simulated dynamics against simulation errors.
Given a simulation algorithm that decomposes the dynamics of the system into many small time steps (e.g., Trotterization), we interweave the simulation with unitary transformations generated by the symmetries of the system (\cref{fig:recipe}).
While these additional unitary transformations increase the gate complexity of the simulation, the error of the simulation can sometimes be reduced by several orders of magnitude, ultimately resulting in a faster quantum simulation.
In addition, depending on the symmetries, the unitary transformations may be implemented using only single-qubit gates, which are considered inexpensive for implementations on near-term quantum computers.

\begin{figure}[b]
\centering
\includegraphics[width=0.45\textwidth]{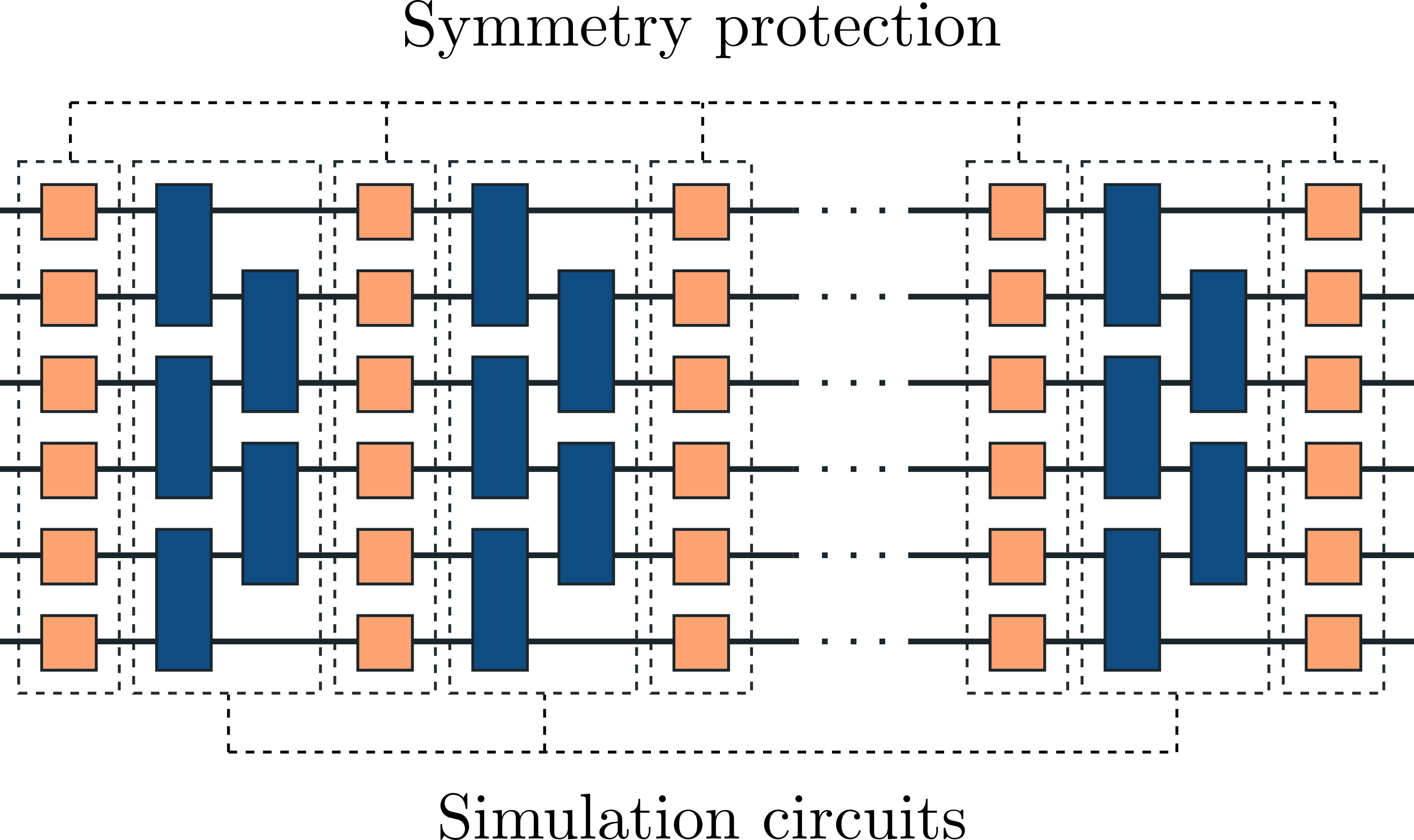}
\caption{For algorithms that simulate the dynamics of quantum systems by decomposing the evolutions into many time steps, we interweave the corresponding simulation circuits (blue) with unitary transformations generated by the symmetries of the systems (orange).
These transformations protect the simulations against errors that violate the symmetries, resulting in faster and more accurate simulations.}
\label{fig:recipe}
\end{figure}

The symmetry protection technique considered in this paper is general and potentially applies to any algorithms that simulate the time evolution of Hamiltonians with symmetries by splitting the evolution into many time segments, including Trotterization and the higher-order product formulas~\cite{suzukiGeneralTheoryFractal1991} and more advanced algorithms such as those based on linear combinations of unitaries~\cite{childsHamiltonianSimulationUsinga,lowWellconditionedMultiproductHamiltonian2019,berrySimulatingHamiltonianDynamics2015c,lowHamiltonianSimulationQubitization2019b}, Lieb-Robinson bounds~\cite{haahQuantumAlgorithmSimulating2018,tranLocalityDigitalQuantum2019}, and randomized compilations~\cite{campbellRandomCompilerFast2019,childsFasterQuantumSimulation2019,campbellRandomCompilerFast2019}.
We also provide evidence that the technique can also protect the simulation against other types of temporally correlated errors, such as the $1/f$ noise commonly found in solid-state devices~\cite{kuhlmannChargeNoiseSpin2013}.

In addition, we draw a connection between the symmetry protection technique and the quantum Zeno effect~\cite{zanardi1999symmetrizing,violaDynamicalDecouplingOpen1999,Facchi04,Khodjasteh08,violaUniversalControlDecoupled1999,ngCombiningDynamicalDecoupling2011,burgarthGeneralizedProductFormulas2019}. 
In particular, the symmetry transformations, when chosen as powers of a unitary, approximately project the error of simulation into the so-called quantum Zeno subspaces, defined by the eigensubspaces of the unitary.
We prove a bound on the accuracy of this approximation, exponentially improving a recent result of Ref.~\cite{burgarthGeneralizedProductFormulas2019}.

The structure of the paper is as follows.
In \cref{sec:general}, we introduce the general technique and provides intuition for the source of error reduction.
In \cref{sec:theory}, we derive a bound on the error of Trotterization under symmetry protection. 
In \cref{sec:app}, we then benchmark the technique in simulating the dynamics of systems with the Heisenberg interactions, including the XXZ Heisenberg model with local disorder that displays a transition between thermalized and many-body localized phases, and in simulating the Schwinger model in the context of lattice field theories.
In particular, we show that interweaving the simulation with random gauge transformations can significantly reduce the probability of a state leaking to outside the physical subspace due to the simulation error, extending the results of Ref.~\cite{stannigelConstrainedDynamicsZeno2014} to digital quantum simulation.
We then demonstrate in \cref{sec:exp} how the technique may protect the simulation against other types of coherent, temporally correlated errors, such as the low-frequency noise typically found in experiments.
Finally, we discuss several open questions in \cref{sec:outlook}.
 
\section{General framework}\label{sec:general}
We consider the task of simulating the time dynamics of a system under a time-independent Hamiltonian $H$.
Let $U_t \equiv \exp(-iHt)$ denote the evolution unitary generated by $H$ for time $t$.
The symmetry protection technique applies to algorithms that simulate $U_t$ by first dividing the evolution into many time steps (also known as \emph{Trotter steps}), and approximate the evolution within each time step by a series of quantum gates. 
Examples of such algorithms include most modern quantum simulation algorithms from the Suzuki-Trotter product formulas~\cite{suzukiGeneralTheoryFractal1991} to algorithms based on linear combinations of unitaries~\cite{childsHamiltonianSimulationUsinga,lowWellconditionedMultiproductHamiltonian2019,berrySimulatingHamiltonianDynamics2015c,lowHamiltonianSimulationQubitization2019b}.
In this paper, we focus our theoretical analysis on the first-order Trotterization algorithm for simplicity~(\cref{sec:theory}) and benchmark the performance of symmetry protection on other algorithms numerically (\cref{sec:qft}).
To be more precise, let $r$ denote the number of steps and $\delta t = t/r$ denote the length of each time step. 
These algorithms then simulate $U_{\delta t}$ by a series of elementary quantum gates $S_{\delta t}$, i.e.
\begin{align} 
	U_{t} = U_{\delta t}^{r} \approx S_{\delta t}^r. \label{eq:rawsim}
\end{align}
The approximation of $U_{\delta t}$ by $S_{\delta t}$ introduces an error that is small for small $\dt$.
However, errors typically accumulate after many Trotter steps, resulting in a total additive error $\norm{U_{t}-S_{\delta t}^r}$ that, in the worse case, scales linearly with the number of Trotter steps $r$ at fixed $\dt$. 
Equivalently, for a fixed total time $t$, to reduce the total error, we would have to decrease the Trotter step size $\dt$, effectively increasing the number of Trotter steps $r$, and thus require more elementary quantum gates to run the simulation.

We refer to the simulation in \cref{eq:rawsim} as the \emph{raw} simulation. 
By exploiting symmetries of the system, we will see that we can substantially reduce the total error $\epsilon$ of the simulation without significantly increasing the gate count, ultimately resulting in faster quantum simulation for the same total error budget. For that, we assume that the Hamiltonian is invariant under a group of unitary transformations, which we denote by $\mathcal S$. Explicitly, we assume that
\begin{align} 
	 \comm{C,H} = 0\quad \forall\ C\in\mathcal S.
\end{align}
The group $\mathcal S$ represents a symmetry of the system.
Instead of simply approximating $U_{\delta t}$ by the circuit $S_{\delta t}$, we ``rotate'' each implementation of $S_{\delta t}$ by a \emph{symmetry transformation} $C_{k}\in \mathcal S$ ($i = k,\dots,r$) so that the approximation in \cref{eq:rawsim} now reads
\begin{align} 
	 U_{t} \approx \prod_{k=1}^r C_k^\dag S_{\delta t} C_k.\label{eq:spsim}
\end{align}
We refer to \cref{eq:spsim} as a \emph{symmetry-protected} (SP) simulation.
The right-hand side in \cref{eq:spsim} represents a circuit that, at first, looks more expensive than \cref{eq:rawsim} due to the additional implementation of the transformations $C_k$. 
However, for the same $r$, the total error in \cref{eq:spsim} could be much smaller than the \cref{eq:rawsim}.
Effectively, to meet the same error tolerance, \cref{eq:spsim} may require a much smaller number of steps $r$, and hence fewer implementations of $S_{\delta t}$, than the raw approximation in \cref{eq:rawsim}.
Moreover, because many symmetries---the gauge symmetries in lattice field theories for example---are spatially local, each $C_k$ only involves a small number of nearest-neighboring qubits and can be implemented easily in most architectures of quantum computers.
Other symmetries, such as the one responsible for the conservation of the total magnetization in the Heisenberg model, are global but may be implemented as a product of only single-qubit gates, which are usually much ``cheaper'' to perform in experiments than their multi-qubit counterparts.

In the remainder of this section, we provide some intuition, using lowest-order arguments, for the error reduction in simulations under symmetry protection.
We later derive rigorous error bounds in \cref{sec:theory}.

\subsection{Lowest-order arguments} \label{sec:Zeno-handwavy}

To build an intuition for the symmetry protection, we consider the {effective Hamiltonian} of the simulation. 
The aim of digital quantum simulation is to simulate the time evolution $e^{-iHt}$ of a Hamiltonian $H$.
Assuming that the simulation errors are coherent, we may end up with the time evolution of a different Hamiltonian, say $H_\eff$, that may be close but not the same as the targeted Hamiltonian $H$:
\begin{align} 
	e^{-iHt} \xrightarrow[]{\text{errors}} e^{-iH_\eff t} = e^{-i(H+V)t}, 
\end{align}
where 
\begin{align} 
	V \equiv  H_\eff - H
\end{align}quantifies the difference between the effective and the desired Hamiltonians~\cite{childsTheoryTrotterError2020}.
We note that the effective Hamiltonian $H_\eff$ typically depends on the time step $\dt$ [See \cref{lem:Htilde}]. 

With $S_\dt = \exp(-iH_\eff \dt)$ in \cref{eq:spsim}, we can rewrite the simulation as
\begin{align} 
	\prod_{k}^r C_k^\dag S_\dt C_k 
	&= \prod_{k=1}^{r}e^{-i C_k^\dag H_\eff C_k \dt} \nonumber\\
	&= \prod_{k=1}^{r}e^{-i (H+C_k^\dag V C_k) \dt}, \label{eq:handwavybeforeBCH}
\end{align}
where we have used the unitarity of $C_k$ to move the unitaries to the exponents and exploited the commutativity $\comm{C_k,H} = 0$ from our assumption to simplify the expression.
Assuming that the error $\norm V$ is small, we can use the Baker-Campbell-Hausdorff (BCH) formula to combine the exponents in \cref{eq:handwavybeforeBCH} (to the leading order):
\begin{align} 
	 \prod_{k=1}^{r}e^{-i (H+C_k^\dag V C_k) \dt}
	 \approx e^{-i \left(H+\frac{1}{r}\sum_{k=1}^r C_k^\dag V C_k\right) t}=e^{-i\overbar{H}_\eff t}.
\end{align}
Compared to the desired evolution $e^{-iHt}$, we can identify the error of the entire simulation (ignoring the error from the BCH approximation for now) as
\begin{align} 
	\overbar V  \equiv \frac{1}{r}\sum_{k=1}^r C_k^\dag V C_k. \label{eq:totalexperror}
\end{align}
Roughly speaking, the error of the entire simulation, given by \cref{eq:totalexperror}, can be interpreted as the average of the error in each step of the simulation.
To illustrate the effect of the symmetry protection, we could imagine $V$ as a vector in the space of operators and $C_k^\dag V C_k$ is a version of the vector rotated around an axis specified by $C_k$.
The total error is then analogous to a walker that, in each step, walks a distance $\norm V$ in the space of operators towards the direction corresponding to~$C_k$~(\cref{fig:walks}). 

\begin{figure}[t]
\centering
\includegraphics[width=0.42\textwidth]{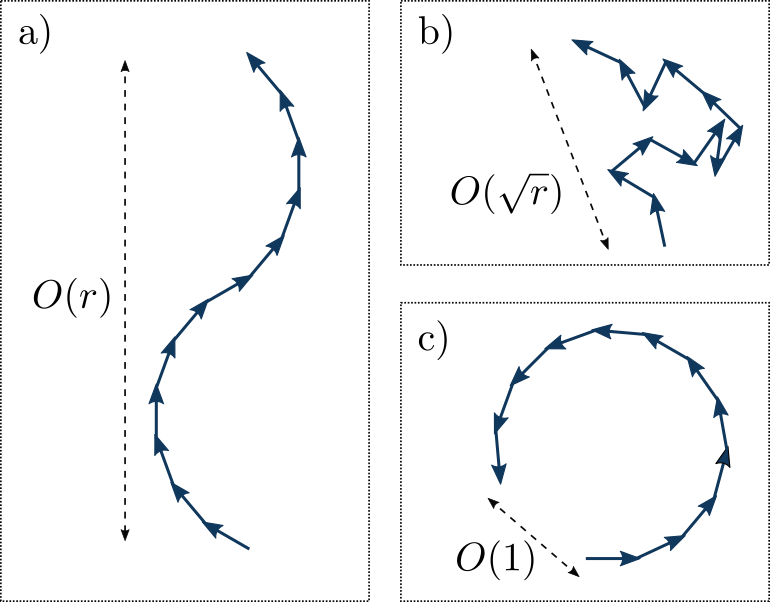}
\caption{The total error of the simulation is analogous to the average distance a walker walks in $r$ steps of the simulation.
In each time step, the walker walks a small distance along a vector representing the error operator in the space of operators.
a) Without any symmetry protection, the walker keeps walking towards almost the same direction, resulting in a total distance that scales linearly with the number of steps $r$, corresponding to the total error scaling as $\O{1}$. 
b) The symmetry transformations make the walker walk in a possibly different direction in every time step. When the direction is uniformly random (see \cref{sec:homo-rand-Heisenberg} and \cref{fig:diff-algo} for an example), the total distance only scales as $\O{\sqrt{r}}$, resulting in the total error scaling as $\O{1/\sqrt{r}}$. 
c) Sometimes, it is possible to design an optimal set of symmetry transformations that makes the walker return to the origin [See \cref{eq:optimalC} for an example], resulting in an $\O{1/r}$ error for the entire simulation.
}
\label{fig:walks}
\end{figure}

Without the symmetry protection (i.e. $C_k = \mathbb I$ for all $k$), the walker keeps walking in the same direction and its total distance after $r$ steps scales as $\O{r}$, resulting in the averaged error $\norm{\overbar V}$ of the same order as~$\norm{V}$.
On the other hand, under the symmetry protection, the walker walks in a possibly different direction in each step, resulting in a smaller total distance (and thus a smaller averaged error.)  

In particular, if the walker in each step walks towards a uniformly random direction in the space of operators (which is sometimes the result of choosing $C_k$ at random), its total distance should only scale as $\O{\sqrt{r}\norm{V}}$ after $r$ steps.
The averaged error $\norm{\overbar V}$ would then scale as $\O{\norm V/\sqrt{r}}$, decreasing with the number of steps of the simulation.
Additionally, if we could design a set of optimal symmetry transformations that makes the walker return to the origin after a fixed number of steps, we would end up with a total distance that does not increase with $r$ and an averaged error $\norm{\overbar V}$ that decreases with $r$ as $\O{\norm V/r}$.
We derive rigorous bounds to support this intuition in \cref{sec:theory}.

The aim of the symmetry protection technique is to choose the symmetry transformations $C_k$ that minimize the error in \cref{eq:totalexperror}.
While each $C_k$ may be chosen independently of the others, we will sometimes focus our attention on a special construction that requires $C_k = C_0^k$ for some $C_0 \in \mathcal S$. 
This choice for the transformations result in a simpler simulation circuit, i.e.
\begin{align} 
 	U_t \approx C_0^{\dag r} (S_\dt C_0)^r, 
 \end{align} 
which corresponds to applying the same symmetry transformation $C_0$ alternatively with the implementations of the simulating circuit $S_\dt$, followed by a final application of $C_0^{\dag r}$ to negate the effect of $C_0$ on the correct evolution.
We could either draw $C_0$ randomly from the symmetry group $\mathcal S$ or infer an optimal choice of $C_0$ from the structure of the error $V$ [See \cref{eq:optimalC} for an example]. 
We analyze the error bounds for the simulation under the protection from this special construction in \cref{sec:theory} and present similar analysis for the general scenario in \cref{sec:pf1}.

It is worth noting that the symmetry transformation $C_0$ introduced above is also analogous to the fast pulses (or ``kicks'') commonly used in quantum control to confine the dynamics of quantum systems~\cite{zanardi1999symmetrizing,violaDynamicalDecouplingOpen1999,Facchi04,Khodjasteh08,violaUniversalControlDecoupled1999,ngCombiningDynamicalDecoupling2011,burgarthGeneralizedProductFormulas2019}.
In fact, we also show in \cref{sec:zeno,sec:ZenoxSP} that a restricted version of the symmetry protection technique is exactly equivalent to frequently applying fast pulses to the systems, resulting in the error being approximately projected onto the so-called quantum Zeno subspaces.
We prove a bound on the error of this approximation, exponentially improving a recent result of Ref.~\cite{burgarthGeneralizedProductFormulas2019}.
This quantum Zeno framework provides an alternative explanation for how quantum simulation can be improved by symmetry protection.

\section{Faster Trotterization by symmetry protection}\label{sec:theory}

In this section, we analyze the effect of the symmetry protection on the total error of the first-order Trotterization algorithm.
Suppose the Hamiltonian $H = \sum_{\mu=1}^L {H_\mu}$ is a sum of $L$ Hamiltonian terms $H_\mu$ such that each $e^{-iH_\mu \dt}$ can be readily simulated on quantum computers.
For readability, we define the following quantities
\begin{align} 
	&\alpha \equiv \sum_{\mu = 1}^{L} \sum_{\nu = \mu+1}^L \norm{\comm{H_\nu,H_\mu}},\\
	&\beta \equiv \sum_{\mu = 1}^{L} \sum_{\nu = \mu+1}^L\sum_{\nu' = \nu}^L \norm{\comm{H_{\nu'},\comm{H_\nu,H_\mu}}},	 
\end{align} 
that depend only on the commutators between the terms of the Hamiltonian.
We will also use the standard Bachmann-Landau big-$O$ and big-$\Theta$ notations in analyzing the asymptotic scalings of the errors with respect to $n,t$, and $r$.
For reference, $\alpha = \O{n}$ and $\beta = \O{n}$ in a system of $n$ nearest-neighbor interacting particles~\cite{childsTheoryTrotterError2020}.

Given a set of symmetry transformations $\mathcal C = \{C_k:k=1,\dots,r\}$, we define 
\begin{align} 
	\overbar{A} \equiv \frac{1}{r} \sum_{k=1}^{r} C_k^\dag A C_k 
\end{align}
as the version of an operator $A$ averaged over the rotations induced by $C_k$.

The first-order Trotterization algorithm approximates $\exp(-i H \dt)$ by
\begin{align} 
	 S_\dt = \prod_{\mu = 1}^L e^{-iH_\mu\delta t},\label{eq:PF1}
\end{align}
where $\prod_{\mu = 1}^L U_\mu \equiv U_L\dots U_2U_1$ is an ordered product.
We define $H_\eff$ as the generator of $S_\dt$, i.e. $S_\dt = \exp(-iH_\eff \dt)$.
We prove the following lemma, providing the existence and the structure of the generator $H_\eff$.
\begin{lemma}\label{lem:Htilde}
For all $\dt$ such that $\beta \dt \leq \alpha$, $2\alpha \dt \leq \norm H$, and $8\dt \norm H\leq 1$, there exists a generator $H_\eff$ for $S_\dt$ and
\begin{align} 
	H_\eff = H - \frac{i}{2} v_0 \dt + \mathcal V(\dt), 
\end{align}
where 
\begin{align}
v_0 \equiv  \sum_{\mu = 1}^{L} \sum_{\nu = \mu+1}^L \comm{H_\nu,H_\mu}, \label{eq:mathcalE_0}
\end{align}
 $\mathcal V(\dt)$ is an operator
satisfying $\norm{\mathcal V(\dt)}\leq \chi \dt^2$ and
\begin{align} 
	\chi\equiv \beta + 32\alpha\norm H.
\end{align}
\end{lemma}

We provide the proof of \cref{lem:Htilde} in \cref{sec:PF1exp_error_proof}.
The essence of \cref{lem:Htilde} is that the error of the simulation, defined as $V \equiv H_\eff - H $, is given by 
\begin{align} 
	V = -\frac{i}{2} v_0 \dt + \O{\chi \dt^2},
\end{align}
and it follows that $\norm{V}\leq \frac{1}{2}\alpha \dt + \chi \dt^2$.

We now consider the effect of protecting the simulation with a set of symmetry transformations $\{C_k:k=1,\dots,r\}$.
Under this symmetry protection, each circuit $S_\dt$ is replaced by
\begin{align} 
	S_\dt \rightarrow C_k^\dag S_\dt C_k 
	= e^{-i C_k^\dag H_\eff C_k \dt}
	&= e^{-i (H + C_k^\dag V C_k) \dt}, 
\end{align}
where we have used $\comm{C_k,H} = 0$ to simplify the expression.
The full simulation becomes
\begin{align} 
	 \prod_{k=1}^r C_k^\dag S_\dt C_k = \prod_{k=1}^re^{-i (H+C_k^\dag V C_k) \dt}.
\end{align}

In the following analysis, we further assume that the symmetry transformations $C_k$ have the form
 $C_k = C_0^k$, where $C_0$ is a symmetry transformation drawn from the symmetry group $\mathcal S$ (We extend these results to general symmetry transformations in Appendix \ref{sec:pf1}.)
Let $\{e^{-i \phi_\mu}:1\leq \mu \leq m\}$ denote the distinct eigenvalues of $C_0$ and
\begin{align} 
 	\overbar H_\eff = H + \frac{1}{r}\sum_{k=1}^r C_k^\dag V C_k = H + \overbar V.
 \end{align}

\begin{lemma} \label{lem:combine}
If $m \geq 2$, we have
\begin{align}
	&\norm{
	\prod_{k=1}^r C_k^\dag e^{-i H_\eff  \dt} C_k
	-e^{-i\overbar{H}_\eff t}}\nonumber\\
	&\leq\frac{2\xi\sqrt{m}(\norm{H}+\norm{V})\norm{V}t^2\log r}{r},\label{eq:Zenoerror}
\end{align}
where
\begin{align} 
	\xi\equiv\max_{\mu\neq\nu}\abs{\sin\left(\frac{\phi_\mu-\phi_\nu}{2}\right)}^{-1} \label{eq:xi}
\end{align}
is the \emph{inverse spectral gap} that depends on the eigenvalues of $C_0$.
\end{lemma}
The proof of \cref{lem:combine} follows from \cref{lem:TrotterZeno2} in \cref{sec:ZenoxSP}.
We note that the bound in \cref{lem:combine} depends on $m$, the number of unique eigenvalues of $C_0$, which could be a constant, e.g. when $C_0$ is generated by local symmetries, or depend on the system size, e.g. when $C_0$ corresponds to generic rotations generated by global symmetries.
We also note that the inverse spectral gap $\xi$ could be large if $C_0$ is nearly degenerate and one should take this effect into account when choosing the unitary $C_0$.

 \Cref{lem:combine} says that, up to the error given in \cref{eq:Zenoerror}, the simulation under the symmetry protection is effectively described by $\overbar{H}_\eff$.
In particular, the total error of the Hamiltonian under the symmetry protection is
\begin{align} 
	\overbar{V}&=  \overbar{H}_\eff - H = \frac{1}{r} \sum_{k=1}^r C_k^\dag V C_k\\
	& =   \frac{-i}{2}\underbrace{\frac{1}{r}\sum_{k=1}^r C_k^\dag v_0 C_k}_{= \overbar {v_0}} \dt  + 
	\underbrace{\frac{1}{r} \sum_{k=1}^r C_k^\dag \mathcal V C_k}_{= \overbar {\mathcal V}},
\end{align}
where we have replaced the expression of $V$ from \cref{lem:Htilde}.
Note that $\norm{\overbar{\mathcal V}} \leq \norm{\mathcal V}$ by the triangle inequality.
Using the identity
\begin{align} 
	\norm{e^{-i\overbar{H}_\eff t} - e^{-iH t}} \leq \norm{\overbar{H}_\eff - H} t = \norm{\overbar{V}} t, \label{eq:diffHidentity}
\end{align}
we arrive at the following bound on the total error of the simulation.

\begin{theorem}[Quantum simulation by symmetry protection]\label{lem:pf1finalbound}
Assuming that $\beta \dt \leq \alpha$, $2\alpha \dt \leq \norm H$, and $8\dt \norm H\leq 1$, the total error of simulation under the symmetry protection from $\{C_k = C_0^k: C_0 \in\mathcal S, k=1,\ldots,r\}$ can be bounded as
\begin{align} 
	\epsilon & \equiv  \norm{\prod_{k=1}^r C_k^\dag S_\dt C_k - e^{-iHt}}
	\nonumber\\&\leq \norm{\overbar{v_0}}\frac{t^2}{2r} + \chi \frac{t^3}{r^2} + \kappa\frac{t^3\log r}{r^2}, \label{eq:PF1bound}
\end{align}
where 
\begin{align} 
    \chi\equiv \beta + 32\alpha\norm H,\quad
	&\kappa \equiv 48 \xi \sqrt{m} \alpha \norm H,  
\end{align}
$m$ is the number of distinct eigenvalues of $C_0$, and $\xi$ is the inverse spectral gap defined in \cref{eq:xi}.
\end{theorem}
The proof of \cref{lem:pf1finalbound} follows immediately from \cref{lem:combine} and \cref{eq:diffHidentity} [See \cref{sec:thm1proof} for the detailed calculations].
The key feature of \cref{lem:pf1finalbound} is that, to the lowest-order in $\frac{t}{r}$, the error scales with $\norm{\overbar{v_0}}$ instead of $\norm{v_0}$.
Since
\begin{align} 
	 \norm{\overbar{v_0}} = \norm{\frac{1}{r}\sum_{k=1}^r  C_k^\dag v_0 C_k}
\end{align}
is generally smaller than $\norm{v_0}$ when $\comm{C_k,v_0}\neq 0$, we expect a smaller simulation error under the symmetry protection.

For demonstration, we consider the simulation of a Hamiltonian $H$ that is a sum of nearest-neighbor interactions on $n$ particles.
It is straightforward to verify that for this Hamiltonian, $\norm H = \O{n}$, $\norm{v_0}\leq \alpha = \O{n}$, $\beta = \O{n}$, and $\chi = \O{n^2}$.
We will also assume that the number of distinct eigenvalues of the $C_0$ is $m = \O{1}$ (corresponding to local symmetries or highly degenerate transformations)
which results in $\kappa = \O{n^2}$.
We will estimate the required number of steps $r$---a good proxy for the gate count~\footnote{In most quantum computer architectures, implementing one Trotter step costs the same amount of gates regardless of the time step $\delta t$.
Here, we also assume that the cost of implementing each symmetry transformation is negligible compared to the simulation circuit. Therefore, given a fixed $n$, the Trotter number $r$ is proportional to the gate count of the simulation.}---for simulations with and without the symmetry protection.

The first scenario corresponds to an unprotected simulation, where $\overbar {v_0} = v_0 $. 
The total error then scales as
\begin{align} 
	\epsilon = \O{\frac{nt^2}{r}} + \O{\frac{n^2t^3\log r}{r^2}}. 
\end{align}
To meet a fixed error tolerance $\epsilon$, we would have to choose the number of steps $r = \Theta(nt^2/{\epsilon})$.

On the other hand, with symmetry protection, we later show that it is sometimes possible to make $\overbar{v_0}$ vanish completely, making the higher order terms the dominant contribution to the total error [See \cref{eq:optimalC} for an example]. 
For nearest-neighbor interactions, the total error is now
\begin{align} 
	 \epsilon = \O{\frac{n^{2}t^3\log r}{r^2}},
\end{align}
which decreases quadratically with $r$. As a result, we only need 
\begin{align} 
	 r = \tilde \Theta\left(\frac{nt^{3/2}}{\sqrt{\epsilon}}\right),
\end{align}
where $\tilde \Theta(\cdot)$ is $\Theta(\cdot)$ up to a logarithmic correction.
Note that this choice of $r$ also satisfies the conditions in \cref{lem:pf1finalbound} when $t/\epsilon>1$.
Compared to the unprotected simulation, the symmetry protection results in a factor of $\sqrt{t/\epsilon}$ improvement in the required number of steps.
At $\epsilon = 0.01$, the improvement in the scaling with $\epsilon$ alone would result in about a factor of ten reduction in the gate count of the simulation.

Finally, we consider a scenario where $\norm{\overbar{v_0}}\propto \norm{v_0}/r^{\gamma}$ for some $\gamma \in (0,1)$.
We provide an example of such a scaling in \cref{sec:homo-rand-Heisenberg}, where drawing the unitary transformations $C_k$ randomly from the symmetry group results in a scaling with $\gamma = 0.5$.
This scaling of $\norm{\overbar{v_0}}$ results in the total error
\begin{align} 
	\epsilon = \O{\frac{nt^2}{r^{1+\gamma}}} + \O{\frac{n^2t^3\log r}{r^2}}. 
\end{align}
Hence, we require 
\[
r = \max\left\{\Theta\left(\left(\frac{nt^2}{\epsilon}\right)^{\frac1{1+\gamma}}\right),\tilde \Theta\left({\frac{nt^{3/2}}{\sqrt{\epsilon}}}\right)\right\},
\] which is again better than the unprotected simulation by a factor of $\min\{(nt^2/\epsilon)^{\gamma/(1+\gamma)},\sqrt{t/\epsilon}\}$.

We recall that in deriving \cref{lem:pf1finalbound}, we have assumed that the symmetry transformations have the form $C_k = C_0^k$ for some $C_0$. We derive in \cref{sec:pf1} a different bound for the general case where each $C_k$ may be chosen independently. 
This general bound, while appearing more complicated, holds the same key feature to the bound in \cref{lem:pf1finalbound}: the total error, to the lowest-order, scales with an averaged version of $v_0$ (under the symmetry transformations) instead of scaling with $\norm{v_0}$.

\section{Applications}\label{sec:app}
In this section, we apply the symmetry protection technique to the simulation of the Heisenberg model (\cref{sec:Heisenberg}) and lattice field theories (\cref{sec:qft}).
In both cases, we show that the symmetry protection results in a significant error reduction and thereby gives faster quantum simulation.

In particular, we use the simulation of the homogeneous Heisenberg model in \cref{sec:homo-rand-Heisenberg} to demonstrate the improvement on the total error scaling as a function of the number of steps $r$ when the simulation is protected by a random set of unitary transformations and by an optimally chosen set.
In \cref{sec:mbl}, we estimate the required number of Trotter steps as a proxy for the gate count in simulating an instance of the Heisenberg model, commonly found in the studies of the many-body localization phenomenon.
Finally, in \cref{sec:qft}, we consider the probability of the state leaking to unphysical subspaces in the digital simulation of the Schwinger model and show that the symmetry protection from the local gauge symmetries can suppress this leakage by a few orders of magnitude.

\subsection{Heisenberg interactions}\label{sec:Heisenberg}
In this section, we use the symmetries in the Heisenberg model to protect its simulation using the first-order Trotterization. 
A Heisenberg model of $n$ spins can be described by the Hamiltonian
\begin{align} 
	H &= \sum_{i=1}^{n-1} \sum_{j=i+1}^n \left(J_{ij}^{(x)} X_i X_j + J_{ij}^{(y)} Y_i Y_j + J_{ij}^{(z)} Z_i Z_j\right) \nonumber\\
	&\qquad\qquad\qquad\qquad\qquad\qquad\qquad+ \sum_{i=1}^{n}h_i Z_i, \label{eq:Heigeneral}
\end{align}
where $X_i, Y_i, Z_i$ are the Pauli matrices acting on site $i$, $J_{ij}^{(x,y,z)}$ represent the interaction strengths between the spins, and $h_i$ correspond to the strengths at site $i$ of an external magnetic field pointing in the $z$ direction.
The Heisenberg model provides a good description for the behavior of magnetic materials in the presence of external magnetic fields.
Depending on several factors, including the signs of the interactions and the dimensions of the system, the Heisenberg model may undergo a quantum phase transition as we increase the strength of the external magnetic field. 
Several important instances of the Heisenberg model includes the homogeneous Heisenberg model ($J^{(x)} = J^{(y)} = J^{(z)}$), the XXZ model ($J^{(x)} = J^{(y)}$) with local disorder, and the Ising model ($J^{(y)} = J^{(z)} = 0$).  
In the following subsections, we will consider two pedagogical instances of \cref{eq:Heigeneral} with SU(2) and U(1) symmetries respectively and demonstrate how the symmetry protection helps reduce the error in simulating the dynamics of these systems even as they move across critical points.

\subsubsection{Homogeneous, random Heisenberg interactions}
\label{sec:homo-rand-Heisenberg}
We first consider a pedagogical toy model where interactions in \cref{eq:Heigeneral} are homogeneous, i.e. $J_{ij}^{(x)} = J_{ij}^{(y)} = J_{ij}^{(z)} = J_{ij}$ for all $1\leq i<j\leq n$, but each $J_{ij}$ is chosen independently at random between $[-1,1]$.
In addition, we assume that $h_i = 0\ \forall i$, i.e. there is no external magnetic field. 
In this case, \cref{eq:Heigeneral} simplifies to
\begin{align} 
	H &=  \underbrace{\sum_{i=1}^{n-1} \sum_{j=i+1}^n J_{ij} X_iX_j}_{\equiv H_X} + \underbrace{\sum_{i=1}^{n-1} \sum_{j=i+1}^n J_{ij} Y_iY_j}_{\equiv H_Y}\nonumber\\
	&+ \underbrace{\sum_{i=1}^{n-1} \sum_{j=i+1}^n J_{ij} Z_iZ_j}_{\equiv H_Z}. \label{eq:Heirand}
\end{align}
The combination of homogeneous interactions and no external magnetic field make \cref{eq:Heirand} invariant under $\mathcal S = \{W^{\otimes n}:W\in \text{SU(2)}\}$, which contains unitaries that---in the Bloch sphere---simultaneously rotate each spin by the same angle.

\begin{figure}[t]
\centering
\includegraphics[width=0.45\textwidth]{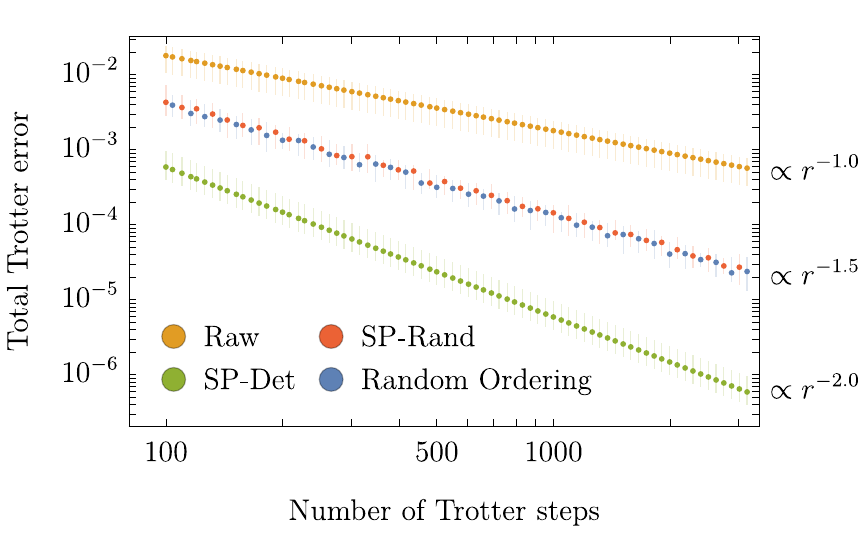}
\caption{The total error in simulating the Hamiltonian \cref{eq:Heirand} at $n = 4$ for a fixed evolution time $t = 1$ as a function of the Trotter number $r$ using four different schemes: the raw first-order Trotterization (``Raw''), the first-order Trotterization protected by a random set symmetry transformation (``SP-Rand''), the first-order Trotterization protected by the optimal set in \cref{eq:optimalC} (``SP-Det''), and the random-ordering scheme in Ref. \cite{childsFasterQuantumSimulation2019} (``Random Ordering'').
We indicate the scalings obtained from power-law fits to the right of the plot.
We repeat the simulation 100 times, each with a different set of randomly generated interactions $J_{ij}$.
The dots correspond to the median of the errors at each value of $r$ and the bars represent the corresponding 25\%-75\% percentiles regions. 
}
\label{fig:diff-algo}
\end{figure}

To simulate the evolution $U_t$ under \cref{eq:Heirand}, we could use the first-order Trotterization to approximate
\begin{align} 
	 U_{t} = \left(e^{-iH\dt}\right)^r \approx \left(e^{-iH_X\dt}e^{-iH_Y\dt}e^{-iH_Z\dt}\right)^r \label{eq:HeiPF1}
\end{align}
by a product of evolutions of individual terms of the Hamiltonian.
The number of Trotter steps $r$ and the time step $\dt = t/r$ determine the error of the simulation.
We refer to this approach as the \emph{raw} Trotterization.
To protect this simulation, we insert unitaries drawn from the symmetry group $\mathcal  S$ in between the Trotter steps, resulting in the simulation
\begin{align} 
 	 U_{t} = \left(e^{-iH\dt}\right)^r \approx \prod_{k=1}^r C_k^\dag  e^{-iH_X\dt}e^{-iH_Y\dt}e^{-iH_Z\dt} C_k,
 \end{align} 
where $\{C_1,\dots,C_r\} \equiv \mathcal C$ is a subset of the symmetry group $\mathcal S$.
Recall that the total error of this symmetry-protected simulation is given by \cref{lem:pf1finalbound}, with the lowest-order error being
\begin{align} 
	\frac{t^2}{2r}\norm{\overbar{v_0}} = \frac{t^2}{2r}\norm{\sum_{k=1}^r C_k^\dag v_0 C_k}, \label{eq:first-order-term-only}
\end{align}
where $v_0  =  \comm{H_Y,H_X} + \comm{H_Z,H_X} + \comm{H_Z,H_Y}$ comes from the leading contribution to the error in one Trotter step.
Different choices of the set $\mathcal C$ lead to different total error of the simulation.

For minimal calculational overhead, we could choose each $C_k$ independently and uniformly at random from $S$ (i.e. $C_k = W_k^{\otimes n}$ where $W_k$ is a Haar random unitary on the single-qubit Bloch sphere.)
The sum in \cref{eq:first-order-term-only} is then the sum of $v_0$, each rotated under a random unitary.
This is analogous to the total error being a random walker that, in each time step, ``walks'' a distance $\norm{v_0}$ in a random direction (See \cref{fig:walks}).
From this analogy, we then expect $\norm{\overbar{v_0}}\propto\norm{v_0}/\sqrt{r}$ (to the lowest-order). 
Therefore, we expect the total error of this scheme to decrease with $r$ as $\O{r^{-3/2}}$ (at fixed total time $t$).

While randomly choosing the unitary transformation set $\mathcal C$ requires little to no knowledge about the error operator $v_0$, one can expect that this choice of $\mathcal C$ is not optimal.
Indeed, by further exploiting the structure of $v_0$,
we can construct a set of transformations $\mathcal C$ that makes \cref{eq:first-order-term-only} vanishes entirely.
One such choice is $C_k = C_0^k$ for $k = 1,\dots,r$, where
\begin{align} 
	C_0 = 
		{U_{H}^{\otimes n}}, \label{eq:optimalC}
\end{align}
and $U_{H}$ is the single-qubit Hadamard matrix.
Alternatively, we could also write
\begin{align} 
	C_k = \begin{cases}
		\mathbb{I} 				&\text{if }k \equiv 0 \mod 2,\\
		{U_{H}^{\otimes n}} 				&\text{if }k \equiv 1 \mod 2,\\
	\end{cases}
\end{align}
for $k = 1,\dots,r$.
Since the Hadamard matrix switches $X \leftrightarrow Z$ and $Y\leftrightarrow -Y$, it is straightforward to verify that \cref{eq:first-order-term-only} vanishes for all even values of $r$.
Therefore, the total error of the simulation is given by the next lowest order in \cref{lem:pf1finalbound}, which scales with $r$ as $\O{1/r^2}$.

In \cref{fig:diff-algo}, we plot the total error of the simulation at $n = 4, t = 1$ as a function of the Trotter number $r$ for the three aforementioned scenarios: the first-order Trotterization without symmetry protection (``Raw''), with symmetry protection from a randomly chosen $\mathcal C$ (``SP-Rand''), and with symmetry protection from the optimal set $\mathcal C$ (``SP-Det'').
The scalings of the errors as functions of $r$ agree remarkably well with our above prediction.
In addition, we also compute the total error using the randomized simulation scheme in Ref.~\cite{childsFasterQuantumSimulation2019}, which decreases the Trotter error by randomizing the ordering of the Hamiltonian terms in between Trotter steps.
Our numerics shows that this scheme performs similarly to the simulation protected by random symmetry transformations, which are both outperformed by the optimal symmetry protection scheme.

\begin{figure*}[t]
\centering
\includegraphics[width=0.45\textwidth]{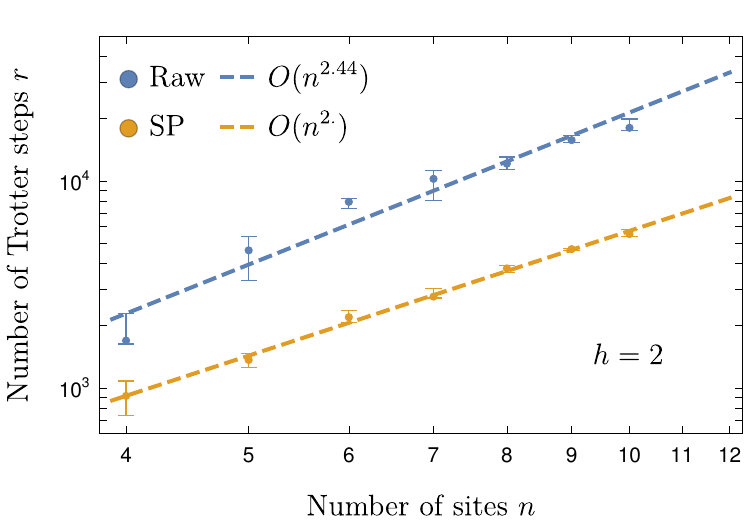}
~
\centering
\includegraphics[width=0.45\textwidth]{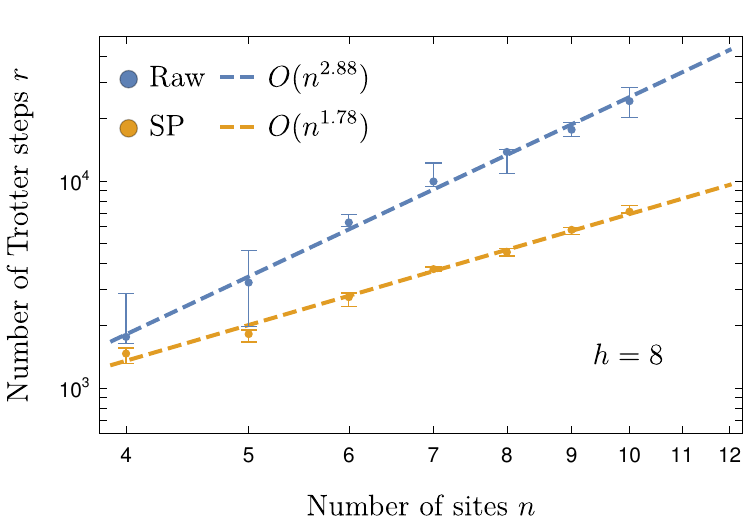}
\caption{The number of Trotter steps required for the simulation of $n$ qubits evolved under \cref{eq:HeiMBL} for time $t = n$ to meet a fixed error tolerance $\epsilon = 0.01$.
We compare this Trotter number of a simulation without any symmetry protection (``Raw'', blue) and a simulation with random symmetry protection (``SP'', orange) at $h = 2$ (left panel) and $h = 8$ (right panel), which correspond to the system being in the ETH and the MBL respectively.
The dashed lines are the linear fits of the data in the log-log scale.
The simulation is repeated 100 times with different instance of the disorder $h_i$.
The dots represent the median of the Trotter number at each $n$ and the error bars correspond to the 25\%-75\% percentile region.
The numerics show that symmetry protecting the simulation reduces the number of Trotter steps, and hence the gate count, by about 2 to 4 times in both the ETH and the MBL phases.
}\label{fig:ETHvsMBL}
\end{figure*}

\subsubsection{Many-body localization}\label{sec:mbl}
The homogeneous Heisenberg interactions without external fields considered in the previous section provides a good testbed for benchmarking the symmetry protection technique. 
In this section, we consider a more physically relevant instance of the Heisenberg model:
\begin{align} 
	H = \sum_{i=1}^{n} \vec \sigma_i \cdot \vec \sigma_{i+1} + \sum_{i=1}^{n}h_i Z_i, \label{eq:HeiMBL}
\end{align}
where we again assume homogeneity for the coupling strengths, but $J_{ij} = 1$ only when $i,j$ are nearest neighbors and $J_{ij} = 0 $ otherwise.
We also adopt the periodic boundary condition and identify the $(n+1)$th qubit as the first qubit.
In addition, we add an external magnetic field with the field strength $h_i$, each chosen randomly between $[-h,h]$.
This model describes homogeneous Heisenberg interactions with a tunable local disorder strength $h$.
At low disorder $h$, the system evolved under \cref{eq:HeiMBL} thermalizes in the long-time limit, in agreement with the Eigenstate Thermalization Hypothesis (ETH). 
However, as $h$ increases, the system transitions to a many-body localized (MBL) phase where it no longer thermalizes
(See \cite{nandkishoreManyBodyLocalizationThermalization2015} for a review of the many-body localization phenomenon.)

To simulate the dynamics of $H$, we again divide the terms of $H$ into groups of mutually commuting terms:
\begin{align} 
	H =  \underbrace{\sum_{i=1}^{n}  X_iX_{i+1}}_{\equiv H_X} + \underbrace{\sum_{i=1}^{n}  Y_iY_{i+1}}_{\equiv H_Y}
	+ \underbrace{\sum_{i=1}^{n} Z_iZ_{i+1} + \sum_{i=1}^n h_i Z_i}_{\equiv H_Z}, \label{eq:HeiMBLdecomposed} 
\end{align}
and use the first-order Trotterization similarly to \cref{eq:HeiPF1}.
To symmetry-protect this simulation, we note that the field term breaks the SU(2) symmetry of the Heisenberg interactions, leaving the system invariant under a U(1) symmetry only.
The symmetry group $\mathcal S = \left\{[\exp(-i\phi Z)]^{\otimes n}:\phi \in [0,2\pi)\right\}$ is generated by the total spin components along the $z$ axis $S_z \equiv \sum_{i=1}^n Z_i$.

While selecting the unitary transformations $C_k$ from this U(1) symmetry is no longer sufficient to completely eliminate the lowest-order error---as we have done in the previous section---we can still expect significantly reduction of the total error due to the symmetry protection and thus a lower gate count for the simulation.
In \cref{fig:ETHvsMBL}, we plot the number of Trotter steps $r$ in simulating the dynamics of \cref{eq:HeiMBL} for time $t = n$ at different values of the disorder $h$ that correspond to the ETH and the MBL phases.
The required numbers of steps are computed at each $n$ by binary searching for the minimum $r$ such that the total error of the simulation does not exceed $\epsilon = 0.01$.
\Cref{fig:ETHvsMBL} shows that protecting the simulation with the U(1) symmetry results in several times reduction in the number of Trotter steps for all values of $n$.
In addition, the Trotter number under symmetry protection also appears to scale better with the system size than in the raw simulation, suggesting an even greater advantage from the symmetry protection for simulating larger systems.

Out of curiosity, we study how the symmetry protection performs as the Hamiltonian moves across the ETH-MBL phase transition.
In \cref{fig:eth-to-mbl}, we plot the required number of steps $r$ in simulating the Hamiltonian of $n = 8$ qubits for time $t = n$ and error tolerance $\epsilon = 0.01$ as we tune the Hamiltonian from the ETH to the MBL phase~\cite{luitzManybodyLocalizationEdge2015}. 
The improvement due to the symmetry protection appears to be unaffected by the phase transition, suggesting that the symmetry protection technique can be useful for future numerical and experimental studies of the transition.

\begin{figure}[t]
\centering
\includegraphics[width=0.45\textwidth]{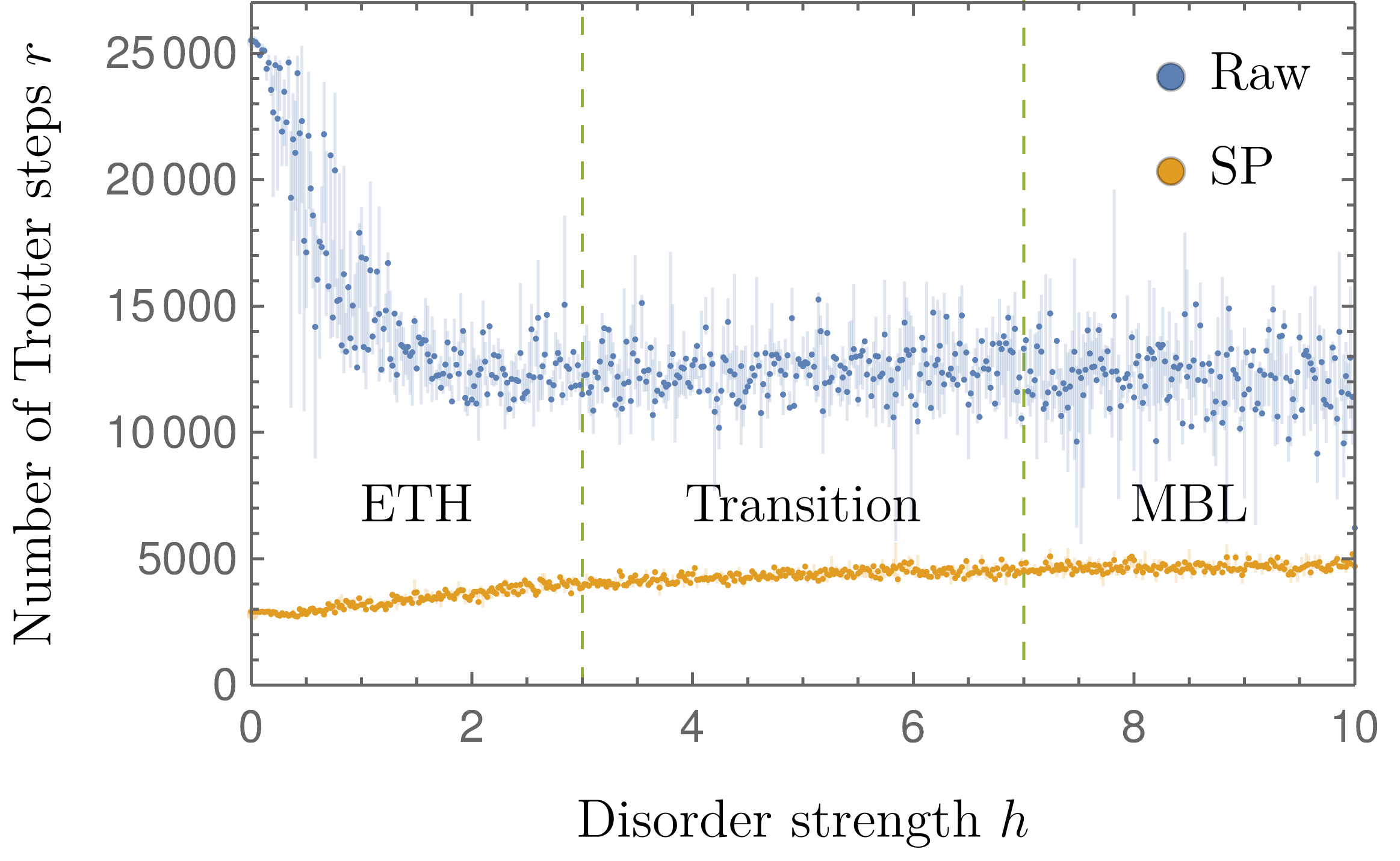}
\caption{The required number of Trotter steps in simulating the Hamiltonian \cref{eq:HeiMBL} of $n = 8$ qubits for time $t = n$ as a function of the disorder strength in an unprotected simulation (``Raw'', blue) and in a symmetry-protected simulation (``SP'', orange).
Each dot represents the median Trotter number over 100 different instances of the random fields. 
The bars correspond to the 25\%-75\% percentile region. 
}
\label{fig:eth-to-mbl}
\end{figure}

\subsection{Simulation of lattice gauge field theories}\label{sec:qft}

Quantum field theories provide another key target for quantum simulation \cite{jordanQuantumComputationScattering2019}. In particular, the quantum simulation of real-time Hamiltonian dynamics, for example scattering processes \cite{jordanQuantumAlgorithmsQuantum2012}, has attracted much attention. An important class of field theories are models with local gauge symmetry, including quantum electrodynamics, chromodynamics, and the Standard Model of particle physics in addition to many condensed matter systems. Substantial effort has gone into the study of analog \cite{haukeQuantumSimulationLattice2013,zoharColdAtomQuantumSimulator2013,davoudiAnalogQuantumSimulations2020} and digital \cite{martinezRealtimeDynamicsLattice2016b,klcoQuantumClassicalComputationSchwinger2018,klcoSUNonAbelianGauge2020,chakrabortyDigitalQuantumSimulation2020,shawQuantumAlgorithmsSimulating2020a} quantum simulation of these models.

In a gauge theory, the system is invariant under a symmetry group which acts separately at each point in space and time (see eg. \cite{kogutHamiltonianFormulationWilson1975} for a review, as well as the lattice Hamiltonian formulation, of these models). This symmetry is fundamentally a redundancy of our description of the physics which we have introduced to give a local description. The Hilbert space $\mathcal{H}$ we use to describe the system contains a subspace $\mathcal{H}_{\rm phys}$ of the physical states, those annihilated by the gauge constraints. For example, in electrodynamics, we have the charge and gauge field degrees of freedom, and the physical states are those annihilated by the Gauss law constraint $\mathcal{G} = \nabla \cdot \mathbf{E} - \rho$, where $\mathbf{E}$ is the electric field operator and $\rho$ is the charge density operator. There are many states in the full Hilbert space $\mathcal{H}$ which do not live in the kernel of $\mathcal{G}$, and these states are not allowed in nature. Although one can in principle work with a description strictly within the physical Hilbert space, it is in general computationally difficult to do the reduction. More importantly, this description would necessarily have a highly spatially non-local set of interactions, a major drawback in practice.

Thus in the simulation of a gauge theory we are faced with a fundamental source of possible errors: what if our dynamics takes us away from the physical Hilbert space? Although the exact Hamiltonian commutes with the gauge constraints, and thus leaves the physical space invariant, an approximate (for example, Trotterized) version of the Hamiltonian may  induce leakage into the unphysical space~\cite{stannigelConstrainedDynamicsZeno2014,shawQuantumAlgorithmsSimulating2020a}. In this section, we apply the symmetry protection technique and use the gauge symmetry itself to protect the simulation against this undesirable leakage \footnote{During the preparation of this manuscript, we learned of related work \cite{lammSuppressingCoherentGauge2020} which provides numerical evidence for the suppression, using gauge transformations, of the experimental drift error in simulating lattice field theories.}.

Explicitly, we consider the one-dimensional Schwinger model~\cite{COLEMAN1976239,PhysRevD.56.55,PhysRevD.66.013002,PhysRevLett.113.091601,shawQuantumAlgorithmsSimulating2020a,chakrabortyDigitalQuantumSimulation2020} consisting of $n$ sites and $n-1$ nearest-neighbor links between the sites. 
We use the formalism outlined in Ref.~\cite{shawQuantumAlgorithmsSimulating2020a}. The Hamiltonian $H = H_0 + H_1$ consists of two terms: 
\begin{align} 
	&H_0 = \sum_{i = 1}^{n-1} F_i^2 -\frac{\mu}{2} \sum_{i=1}^n (-1)^i  Z_i,\label{eq:Schwinger_H0}\\
	&H_1 = x \sum_{i=1}^{n-1}\bigg[
	\frac14(U_i+U_i^\dag)(X_iX_{i+1}+Y_iY_{i+1})\nonumber\\
	&\qquad\qquad+\frac i4(U_i-U_i^\dag)(X_iY_{i+1}-Y_iX_{i+1}) 
	\bigg] , 
\end{align}
where 
\begin{align} 
	&F_i = \sum_{j=-\Lambda}^{\Lambda-1} j \ket{j}_i\bra{j}_i,\\
	&U_i = \sum_{j=-\Lambda}^{\Lambda-2} \ket{j+1}_i\bra{j}_i + \ket{-\Lambda}_i\bra{\Lambda-1}_i,
\end{align}
and $\mu,x$ are positive constants.
Here, $H_0$ describes the on-site and on-link terms, $H_1$ describes the site-link interaction, and 
$F_i$ is the electromagnetic field operator for the link that connects the $i$th and $(i+1)$th particles. 
We note that while the second term in \cref{eq:Schwinger_H0} sometimes appears in the literature without the minus sign (see for example Ref.~\cite{chakrabortyDigitalQuantumSimulation2020}), this discrepancy is the result of different conventions for mapping between fermions and spins and does not have any physical consequences.
In a simulation, we have to put a cutoff $\Lambda$ specifying the maximum excitation number for the bosonic degree of freedom on a given link.

The Hamiltonian is subjected to local symmetries generated by the gauge operators:
\begin{align} 
	\mathcal G_i = F_i - F_{i-1} - Q_i,\label{eq:gaugeop} 
\end{align}
where $Q_i =  \frac12 \left[-{Z_i} + {(-1)^i}\right]$ counts the electric charge at site $i$.
In particular, only states $\ket{\psi}$ that satisfy $\mathcal G_i = 0$ for all $i$ are considered physical.

The physical states form a subspace $\mathcal{H}_{\text{phys}}$ which can be constructed from the kernels of the gauge operators:
\begin{align} 
	\mathcal{H}_{\text{phys}} \equiv \medcap_i \text{Ker}(\mathcal G_i), 
\end{align}
where $\text{Ker}(\mathcal G_i) = \{\ket{\phi}:\mathcal G_i\ket{\phi} = 0\}$ is the kernel of $\mathcal G_i$.

Due to various errors, an initially physical state may leak to unphysical subspace during the simulation.
Formally, we define the leakage of a state $\ket{\psi(t)}$ at time $t$ as
\begin{align} 
	1 - \abs{\bra{\psi(t)}\Pi_0\ket{\psi(t)}}, 
\end{align}
where $\Pi_0$ is the projector onto the physical subspace $\mathcal{H}_{\text{phys}}$.

To simulate $e^{-iH\dt}$ for a small time $\dt$, we first decompose it into $e^{-iH_0\dt}e^{-iH_1\dt}$ using the first order Trotterization. 
Since both $H_0,H_1$ commute with $\mathcal G_i$, this decomposition respects the gauge symmetries and does not result in leakage from the physical subspace.
However, to simulate the evolution under $H_1$, we need to further decompose it into elementary quantum gates.
For that, we follow the steps in Ref.~\cite{shawQuantumAlgorithmsSimulating2020a} and write
\begin{align} 
	U_i + U_i^\dag = A_i + \tilde A_i, 
\end{align}
where $A_i = \mathbb I\otimes \dots\otimes \mathbb I\otimes X$ and $\tilde A_i = U_i^\dag A_i U_i$.
Similarly,
\begin{align} 
	 i (U_i-U_i^\dag) = B_i+ \tilde B_i,
\end{align}
where $B_i = \mathbb I\otimes \dots\otimes \mathbb I\otimes  Y$ and $\tilde B_i = U_i^\dag B_i U_i$.
This representation allows us to decompose the evolution
\begin{align} 
	e^{-iH_0\dt}e^{-iH_1\dt} &\approx S_\dt \equiv e^{-iH_0\dt} \nonumber\\  
	\cdot &\prod_{i} e^{-\frac{1}{4}ix\dt A_i X_iX_{i+1}} 
	e^{-\frac{1}{4}ix\dt \tilde A_i X_iX_{i+1}} \nonumber\\
	\cdot & e^{-\frac{1}{4}ix\dt A_i Y_iY_{i+1}} 
	e^{-\frac{1}{4}ix\dt \tilde A_i Y_iY_{i+1}}\nonumber\\
	\cdot&
	e^{-\frac{1}{4}ix\dt  B_i X_iY_{i+1}}
	e^{-\frac{1}{4}ix\dt \tilde B_i X_iY_{i+1}}\nonumber\\
	\cdot& e^{+\frac{1}{4}ix\dt  B_i Y_iX_{i+1}}
	e^{+\frac{1}{4}ix\dt \tilde B_i Y_iX_{i+1}},\label{eq:Schwingerdecomposed}
\end{align}
into a product of three-qubit gates that can be readily implemented on quantum computers~\cite{shawQuantumAlgorithmsSimulating2020a}.
Note that the cost of simulating $e^{-\frac{1}{4}ixt \tilde A_i X_iX_{i+1}}$ is that of approximating $e^{-\frac{1}{4}ixt A_i X_iX_{i+1}}$, plus the cost of implementing $U_i$, $U_i^\dag$:
\begin{align} 
	 e^{-\frac{1}{4}ix\dt \tilde A_i X_iX_{i+1}} = U_i^\dag e^{-\frac{1}{4}ix\dt A_i X_iX_{i+1}} U_i.
\end{align}
The entire raw first-order Trotterization simulation of $e^{-iHt}$ becomes
\begin{align} 
	e^{-iHt} \approx S_\dt^r. 
\end{align}
Similarly to the Heisenberg model, we could protect this simulation by interweaving the Trotter steps with symmetry transformations of the system:
\begin{align} 
	 e^{-iHt} \approx \prod_{k=1}^r C_k^\dag S_\dt C_k, 
\end{align}
where $C_k$ are generated by the gauge operators in \cref{eq:gaugeop}.
Specifically, we choose 
\begin{align} 
	C_k = \prod_{i=1}^n  \exp(-i \phi_{k,i} \mathcal G_i) 
\end{align}
for some angles $\phi_{k,i}$.

\begin{figure}[t]
\centering
\includegraphics[width=0.45\textwidth]{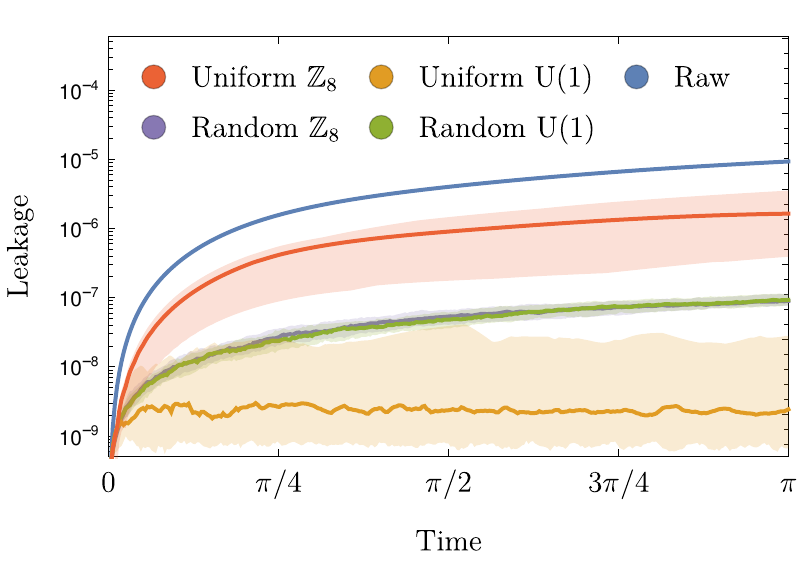}
\caption{The probability for the final state to leak outside the physical subspace due to Trotter errors in simulating the Schwinger model. We consider simulations without symmetry protection (blue) and with symmetry protection under different schemes: uniform sets of transformations drawn from $\mathbb Z_8$ (red) and U(1) (orange) and random sets of transformations drawn from $\mathbb Z_8$ (purple) and U(1) (green). The purple and green areas overlap each other almost completely. The dots correspond to the median and the shaded areas correspond to the 25\%-75\% percentile of 100 repetitions.}
\label{fig:schwinger-leakage}
\end{figure}

Since we truncate the spectrum of each bosonic link to $[-\Lambda+1,\Lambda]$, the transformations $C_{k}$ in general commute with the Hamiltonian of the system only if we choose $\phi_{k,i} = m_{k,i} \pi/\Lambda$, where $m_{k,i}$ are integers.
These transformations effectively form a $\mathbb Z_{2\Lambda}$ symmetry of the truncated Hamiltonian~\cite{kuhnQuantumSimulationSchwinger2014,ercolessiPhaseTransitionsGauge2018a}.
However, the U(1) symmetry can be recovered by assuming a vanishing background field and choosing a large enough cutoff level $\Lambda$ such that, in the physical subspace, the bosonic links never ``see'' the cutoff. 
More rigorously, if $\Lambda > n/2+1$, the transformations $C_k$ commute with $\Pi_0 H \Pi_0$, where $\Pi_0$ is the projection onto the physical subspace $\mathcal H_{\text{phys}}$, for all angles $\phi_{k,i}\in [0,2\pi)$.

In \cref{fig:schwinger-leakage}, we plot the leakage outside the physical subspace due to the Trotter error during simulations with and without symmetry protection. 
Specifically, we simulate the evolution of the ground state of the Schwinger model with 4 sites and 3 links at $x = 0.6$, $\mu = 0.1$, $\dt = 0.01$, and $\Lambda = 4$.
This choice of $\Lambda$ ensures that the Hamiltonian has a $\mathbb Z_8$ symmetry in general and a U(1) symmetry when restricted to the physical subspace.
We consider two choices of the angles $\phi_{k,i}$: $\phi_{k,i} = k\phi_{1,i}$ (``Uniform''), for some randomly chosen $\phi_{1,i}$, and $\phi_{k,i}$ chosen independently at random for each $k$ (``Random''). 
We repeat the simulation 100 times, each with a different choice of the angles.

\Cref{fig:schwinger-leakage} shows that the symmetry protection can reduce the leakage to the unphysical subspace by several orders of magnitude compared to a raw simulation.
While the leakage builds up in a raw simulation, the uniform choice of the transformations from the $U(1)$ symmetry results in bounded leakage during the entire simulation. 
This feature resembles the optimal symmetry protection discussed in \cref{sec:homo-rand-Heisenberg} for the Heisenberg models, where the symmetry protection suppresses the simulation error nearly completely.
Different choices of the symmetry transformations also affect performance of the scheme differently. 
While the random choices of transformations from $\mathbb Z_8$ and U(1) have the same effect on the leakage, the uniform choice of transformations from $\mathbb Z_8$ performs significantly worse than the U(1) counterpart.
This discrepancy is likely because we have only eight choices for the $\mathbb Z_8$ symmetry transformations, whereas with the U(1) symmetry the number of choices is theoretically infinite. Effectively, the symmetry group $\mathbb Z_8$ has less freedom and, therefore, is less effective in averaging out the simulation error than U(1).

While our analysis in \cref{sec:theory} focuses on the application to the first-order Trotterization algorithm, it is clear from the analysis that the symmetry protection will suppress any simulation errors that violate the symmetries of the system, including errors from more advanced algorithms. 
To support this claim, we provide in \cref{fig:other-algorithms} numerical evidence of the symmetry protection suppressing the leakage to unphysical subspace in simulating the Schwinger model using the second-order Suzuki-Trotter formula, the fourth-order Suzuki-Trotter formula~\cite{suzukiGeneralTheoryFractal1991}, and a multi-product formula implemented via a linear combination of unitaries~\cite{childsHamiltonianSimulationUsinga}.

\begin{figure}[t]
\centering
\subfloat[Second-order Suzuki-Trotter \label{fig:PF2}]{\includegraphics[width=0.37\textwidth]{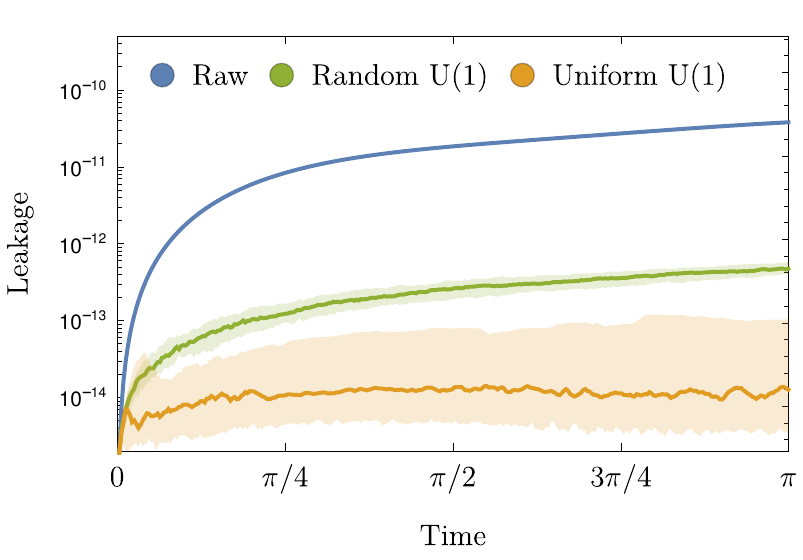}}

\subfloat[Fourth-order Suzuki-Trotter \label{fig:PF4}]{\includegraphics[width=0.37\textwidth]{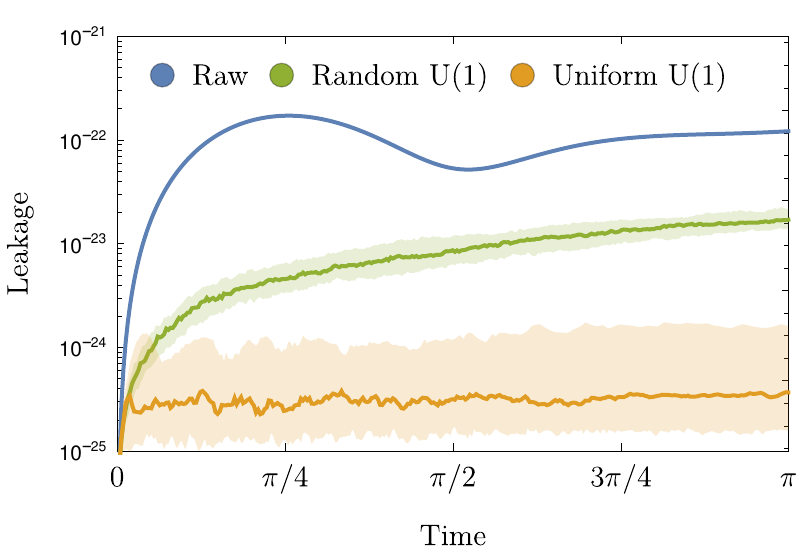}}

\subfloat[Multi-product formula via LCU \label{fig:MultiPF}]{\includegraphics[width=0.37\textwidth]{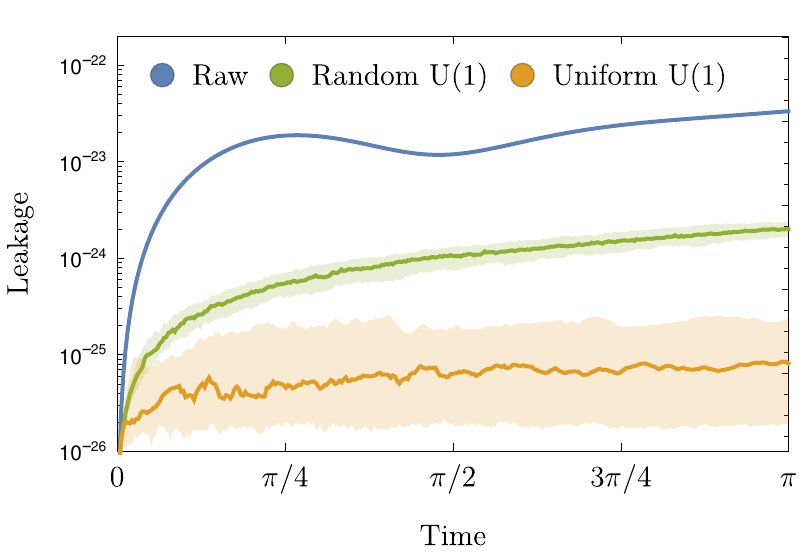}}
\caption{The leakage to the unphysical subspace as a function of time in simulating the Schwinger model using advanced algorithms. We consider a raw simulation (blue), a simulation protected by a random set of transformations drawn from the U(1) symmetry group (green), and a simulation protected by a uniform set of transformations (orange). The solid dots correspond to the median of 100 repetitions and the shaded area corresponds to the 25\%-75\% percentile.  }
\label{fig:other-algorithms}
\end{figure}

Given a Hamiltonian $H = \sum_{\nu =1}^L H_\nu$ being a sum of $L$ terms, the second-order Suzuki-Trotter formula simulate the time evolution $e^{-iH\dt}$ by
\begin{align} 
	P_2(\dt) = \prod_{\nu = 1}^{L} e^{-iH_\nu \frac{\dt}{2}}\cdot\prod_{\nu = L}^{1} e^{-iH_\nu \frac{\dt}{2}},  
\end{align}
which is correct up to an $\O{\dt^3}$ error.
The formula can be generalized to any even order $p\geq 2$ through a recursive construction~\cite{suzukiGeneralTheoryFractal1991}:
\begin{align} 
 	P_p(\dt) = P_{p-2}(\kappa_p \dt)^2P_{p-2}((1-4\kappa_p) \dt)P_{p-2}(\kappa_p \dt)^2, 
 \end{align} 
where $\kappa_p = 1/(4-4^{1/p})$.
The $p$th-order formula approximates $e^{-iH\dt}$ up to an error $\O{\dt^{p+1}}$.
Given a small $\dt$, the formulas can be made arbitrarily accurate by increasing $p$ at the cost of increasing the gate count exponentially with $p$.

In contrast, multi-product formulas~\cite{chinMultiproductSplittingRungeKuttaNystrom2010a} enable the construction of any $p$th-order approximations using only $\poly(p)$ quantum gates by approximating the time evolution by sums of product formulas.
Asymptotically, the gate counts of the multi-product formulas have polylogarithmic dependence on the inverse of the error tolerance.
Therefore, when used as a subroutine in the Lieb-Robinson-bound-based algorithm~\cite{haahQuantumAlgorithmSimulating2018}, the multi-product formulas also result in asymptotically optimal gate counts, up to polylogarithmic corrections, in simulating geometrically local systems. 
Because a sum of product formulas is generally not unitary, it must be implemented using techniques such as linear combinations of unitaries (LCU)~\cite{childsHamiltonianSimulationUsinga}, which encodes the multi-product formula into a unitary acting in a larger Hilbert space.
Here, we will simulate the Schwinger model using a multi-product formula constructed by Childs and Wiebe~\cite{childsHamiltonianSimulationUsinga}:
\begin{align} 
  	 M(\dt) = \frac{16}{15} P_2(\dt/4)^4 - \frac{1}{15} P_2(\dt), \label{eq:multiPF}
\end{align}  
which is a linear combination of two second-order product formulas.

\Cref{fig:other-algorithms} plots the leakage to the unphysical subspace during the simulation at $n = 4, x = 0.6,\mu = 0.1,\dt = 0.01$, and $\Lambda = 4$ using the second-order Suzuki-Trotter formula, the fourth-order Suzuki-Trotter formula, and the multi-product formula [\cref{eq:multiPF}] with and and without symmetry protection. 
We implement the multi-product formula using LCU and an additional ancillary qubit.
For the considered algorithms, the numerics show similar features to \cref{fig:schwinger-leakage}, where the symmetry protection suppresses the leakage by several orders of magnitude and, in particular, the uniform choice of transformations results in bounded errors throughout the simulation.
The figure therefore demonstrates the generality of our approach in protecting digital quantum simulations against errors that violate symmetries of the target system.
We note that the dips in the leakage of the raw simulations are likely due to the small system size considered in the simulations.

\begin{figure*}[t]
\centering
\includegraphics[width=0.75\textwidth]{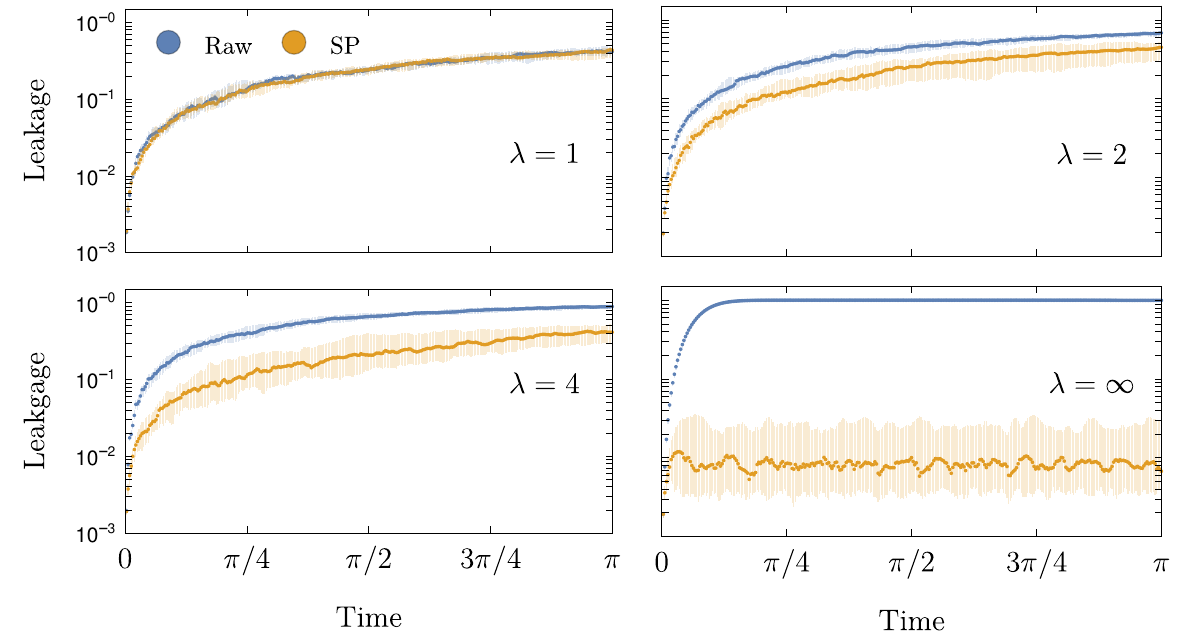}
\caption{The leakage probability due to experimental noise as a function of time at different values of the correlation length $\lambda$.
The simulation is repeated 100 times with different instances of the experimental noise.
The solid dots represent the median of the leakage and the bars correspond to the 25\%-75\% percentile regions.}
\label{fig:leakagevscorrlength}
\end{figure*}
\section{Additional protection against experimental errors}\label{sec:exp}

So far, we have demonstrated that symmetries in quantum systems can be used to suppress the simulation error of the Trotterization algorithm.
In this section, we discuss how the technique may also protect the simulation against other types of error, including the experimental errors that may arise in the implementation of Trotterization.

In our earlier derivation, we show that the lowest-order contribution to the total error is
\begin{align} 
	\norm{\overbar v_0}=\frac{1}{r}\norm{\sum_{k=1}^r C_k^\dag  v_0  C_k},
\end{align}
where $v_0$ is the lowest-order error from the simulation algorithm.
This derivation applies equally well for the case when the error $v_0$ comes from sources other than the approximations in the simulation algorithms.

However, in our analysis, we require that $v_0$ remains the same for different steps of the simulation.
In other words, the error $v_0$ for different Trotter steps are correlated in time.
In particular, an error with temporal correlation lengths being longer than the time step $\dt$ would enable us to choose the symmetry transformations such that the errors from several consecutive steps interfere destructively.
Therefore, we expect the symmetry protection technique to help reduce low-frequency noises, such as the $1/f$ noise typically found in solid-state qubit systems.

We provide numerical evidence for this argument by adding temporally correlated errors to the simulation of the Schwinger model. 
Specifically, after each step $k$ of the simulation, we apply single-qubit rotations $\exp(-i \eta\ \vec\sigma \cdot \hat n_k)$ on the system, where $\eta = 0.01$ is a small angle, around a random axis $\hat n_k$.
These rotations mimic the effect of a depolarizing channel and violate the gauge symmetries, resulting in the state leaking to the unphysical subspace.
To impart temporal correlations into this noise model, we choose the random unit vectors $\hat n_k$ again only after every $\lambda$ consecutive Trotter steps.
The parameter $\lambda$ therefore plays the role of the correlation length of the noise.

In \cref{fig:leakagevscorrlength}, we plot the probability that the state leaks to unphysical subspace (due to the simulation error) as a function of time for several values of the correlation length $\lambda$.
To study the effect of the symmetry protection technique on the added experimental noise, we use the fourth-order Trotterization in the simulation to suppress the algorithm error, making the added noise the main contributor to the leakage observed in \cref{fig:leakagevscorrlength}.
As expected, at $\lambda = 1$, the experimental error varies too fast between Trotter steps and is immune against the symmetry protection technique.
However, the technique begins to suppress the experimental error as soon as the noise becomes temporally correlated ($\lambda>1$) and becomes more effective as the correlation length $\lambda$ increases.
Even at $\lambda = 4$, we have managed to reduce error by about an of magnitude.

\section{Discussion \& Outlook}\label{sec:outlook}
In this paper, we propose a general technique to suppress the error of quantum simulation using the symmetries available in quantum systems, ultimately resulting in faster digital quantum simulation.
We have analyzed the technique when applied to the Trotterization algorithm and derived bounds on the total error of the simulation under symmetry protection. 
The bound provides insights for choosing the set of unitary transformations that optimally suppress the simulation error.
We then benchmarked the technique in simulating the Heisenberg model and lattice field theories.
Both examples showed that the symmetry protection results in significant reduction in the total error, and thus the gate count, of the simulation.  
Finally, we argue that the technique can also protect digital quantum simulation against temporally correlated noise in experiments.

An immediate future direction is to generalize the analysis in this paper to more advanced quantum simulation algorithms, such as the higher-order Suzuki-Trotter formulas~\cite{suzukiGeneralTheoryFractal1991}, the truncated Taylor series~\cite{berrySimulatingHamiltonianDynamics2015c}, or qubitization~\cite{lowHamiltonianSimulationQubitization2019b}.
We emphasize that our approach induces destructive interference between the errors from different steps of the simulation and, therefore, should suppress errors that violate the symmetries of the target system, regardless of the sources of the errors.
However, the optimal choice of the symmetry transformations depends on the exact error structure in each step of the simulation.
Since the error structures of more advanced algorithms are typically more complicated than the first-order Trotterization, it is more difficult to infer the set of symmetry transformations that optimally protects the simulation.
Nevertheless, extensive analytical and numerical studies of the effectiveness of the technique for protecting these advanced algorithms, especially when applied to the simulations of various physically relevant systems, such as the lattice field theories~\cite{martinezRealtimeDynamicsLattice2016b,klcoQuantumClassicalComputationSchwinger2018,klcoSUNonAbelianGauge2020}, or the electronic structures~\cite{Poulin15,babbushLowDepthQuantumSimulation2018,mottaLowRankRepresentations2018b,McArdle20}, would be useful for the long-term development of digital quantum simulation.

When the error structure of the algorithm is not readily available, an alternative promising approach for optimizing the set of symmetry transformations is to parameterize the transformations, variationally minimize the error of the first few simulation steps, and apply the same set of transformations repeatedly for the rest of the simulation \cite{cirstoiu2019variational}.
Understanding when such a variational approach can suppress the error in a long simulation could provide a path towards a scalable symmetry protection with a minimal calculation overhead.

In addition, our analysis in this paper focuses primarily on the error of the simulation algorithm under the symmetry protection in the full Hilbert space.
It would be interesting to, for example, build upon the recent result of Ref.~\cite{sahinogluHamiltonianSimulationLow2020} and analyze the symmetry-protected simulation error in a low-energy subspace. 

Lastly, we would like to note that, although our analysis focuses on digital quantum simulation, we expect the symmetry protection technique to apply equally well for analog quantum simulation and classical simulation of the dynamics of quantum systems.

\begin{acknowledgments}
	We thank Ryan Babbush, Andrew Childs, Su-Kuan Chu, Zohreh Davoudi, Jens Eisert, Mária Kieferová, Natalie Klco, Hank Lamm, Guang Hao Low, Nhung Nguyen, Alexander Shaw, and Nathan Wiebe for helpful discussions. Partial support for this research is provided by the Princeton Center for Complex Materials, a MRSEC supported by NSF Grant DMR No. 1420541 and by the U.S. Department of Energy, Office of Science, Office of Advanced Scientific Computing Research, Quantum Algorithms Teams and Quantum Testbed Pathfinder programs (Award No. DE-SC0019040). M.C.T. and Y.S. acknowledge additional funding by ARO MURI, NSF (Grant No. CCF-1813814) and Accelerated Research in Quantum Computing (Award No. DE-SC0020312) program. M.C.T also acknowledges DoE BES Materials and Chemical Sciences Research for Quantum Information Science program (Award No. DESC0019449), NSF PFCQC program, AFOSR, AFOSR MURI, ARL CDQI, and NSF PFC at JQI. Y.S. is supported by the Google Ph.D. Fellowship program. He also acknowledges the National Science Foundation RAISE-TAQS 1839204 and Amazon Web Services, AWS Quantum Program. The Institute for Quantum Information and Matter is an NSF Physics Frontiers Center PHY-1733907. Fermilab is operated by Fermi Research Alliance, LLC, under Contract No. DEAC02-07CH11359 with the US Department of Energy. The authors acknowledge the University of Maryland supercomputing resources (http://hpcc.umd.edu) made available for conducting the research reported in this paper.  
\end{acknowledgments}

\bibliography{sym-sim,zotero-sym-sim-do-not-edit}

\newpage

\begin{widetext}
\appendix

\section{Faster convergence of quantum Zeno effect}\label{sec:zeno}

Using symmetries to protect quantum simulations has previously been explored in the context of the quantum Zeno effect:
undesirable errors from the simulation can be suppressed by constantly measuring the system in an appropriate basis~\cite{zanardi1999symmetrizing,stannigelConstrainedDynamicsZeno2014,strykerOraclesGaussLaw2019}.
However, measurements are costly in most available quantum computers and therefore often only performed once at the end in simulations on quantum computers.
An alternative approach commonly used in quantum control is to frequently apply fast pulses, or ``kicks'', to the system during the experiments.
In the high frequency limit, these kicks confine the dynamics of the system to the so-called quantum Zeno subspaces defined by the spectral decomposition of the kicks~\cite{zanardi1999symmetrizing,violaDynamicalDecouplingOpen1999,Facchi04,Khodjasteh08,violaUniversalControlDecoupled1999,ngCombiningDynamicalDecoupling2011,burgarthGeneralizedProductFormulas2019},
effectively realizing the quantum Zeno effect without measuring the systems.

In this section, we derive a concrete bound on the rate at which the effective Hamiltonian of a frequently kicked system converges to its projection to the Zeno subspaces. This bound exponentially improves a recent result of Burgarth, Facchi, Gramegna, and Pascazio~\cite{burgarthGeneralizedProductFormulas2019}. Interestingly, our proof makes use of a tight analysis of Trotter error \cite{childsTheoryTrotterError2020}, suggesting a deep connection between quantum simulation and quantum Zeno effect.

The aim of quantum control is to confine the dynamics of a system evolving under a Hamiltonian $G$ into the subspaces specified by a set of projectors:
\begin{align} 
	\mathcal P = \{P_\mu\}. 
\end{align}
One approach is to repeatedly measure the system in the basis corresponding to $\mathcal P$ throughout the evolution.
These measurements results in the quantum Zeno effect: the dynamics of the system is confined to the subspaces corresponding to the projectors $P_\mu$.
Alternative to measuring the system, one could periodically ``kick'' the system~\cite{burgarthGeneralizedProductFormulas2019} with a unitary
\begin{align} 
	U_{\text{kick}} = \sum_{\mu} e^{-i\phi_\mu} P_{\mu},\label{eq:Ukickdef}
\end{align}
where $\phi_\mu$ is chosen such that $\phi_\mu \neq \phi_\nu\mod 2\pi$  for all $\mu \neq \nu$.

Suppose the total evolution time is $t$ and we apply the kick every $\dt = t/r$ where $r$ is an integer.
The dynamics of the system becomes
\begin{align} 
	(U_\kick^{\dag})^r \left(e^{-iGt/r}U_\kick\right)^r,\label{eq:Zenodynamics}
\end{align}
where $(U_{\kick}^{\dag})^r$ is added at the end of the sequence to undo the evolution generated by the $r$ applications of $U_{\kick}$.
In the limit $r\rightarrow\infty$, the dynamics of the system again exhibits the quantum Zeno effect
\begin{equation}
	U_{\text{kick}}^{\dagger r}\left(e^{-i\frac{t}{r}G}U_{\text{kick}}\right)^r\rightarrow e^{-itG_{\text{Zeno}}},
\end{equation}
where
\begin{equation}
	G_{\text{Zeno}}\equiv\sum_{\mu=1}^{m}P_\mu GP_\mu,
\end{equation}
is the projection of $G$ onto the subspaces defined by the spectral decomposition of $U_{\text{kick}}$.
In other words, the kicks effectively confine the dynamics of the system to the subspaces defined by the projectors $P_\mu$~(See \cref{fig:zeno}).

\begin{figure}[t]
\centering
\includegraphics[width=0.42\textwidth]{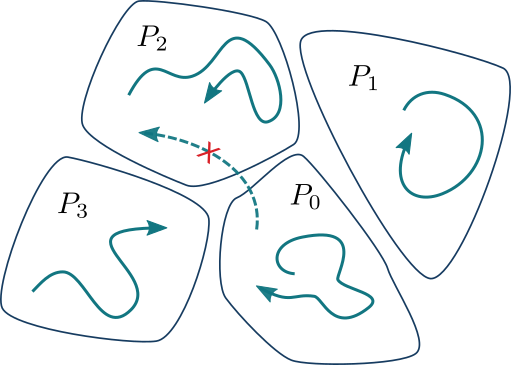}
\caption{The frequent kicks confine the dynamics of the system (solid arrows) to the so-called quantum Zeno subspaces, defined by the projectors $P_\mu$ in the spectral decomposition of the kicks $U_{\text{kick}} = \sum_{\mu} e^{-i\phi_\mu} P_{\mu}$.
In particular, the kicks suppress the probability for the system to travel between the subspaces (dashed arrow).
By generating the kicks from the symmetries of the system, we can target the simulation error---the sole contributor to possible violations of the symmetries in an ideal simulation---for suppression.}
\label{fig:zeno}
\end{figure}

Ref.~\cite[(A.30)]{burgarthGeneralizedProductFormulas2019} derived the following bound on the convergence rate with explicit dependence on all parameters of interest
\begin{align}
	&\norm{U_{\text{kick}}^{\dagger r}\left(U_{\text{kick}}e^{-i\frac{t}{r}G}\right)^r - e^{-itG_{\text{Zeno}}}}
	\leq\frac{\xi m^2\norm{G}t(1+2e^{m\norm{G}t})}{r},\label{eq:old_zeno_bound}
\end{align}
where $m$ is the number of projectors and
\begin{align} 
	\xi\equiv\max_{\mu\neq\nu}\abs{\sin\left(\frac{\phi_\nu-\phi_\mu}{2}\right)}^{-1}  \label{eq:CZeno_def}
\end{align}
is the inverse spectral gap.
Unfortunately, this bound has exponential dependence on $m$, $\norm{G}$, and $t$, which, in particular, suggests that we have to increase the number of kicks $r$ exponentially with the evolution time of the system and therefore may be impractical in many applications. 
In \cref{thm:Zeno}, we prove a different bound that exponentially improves the bound in Ref.~\cite{burgarthGeneralizedProductFormulas2019} in terms of $m$, $\norm{G}$, and~$t$.

\begin{theorem}[Faster convergence of quantum Zeno effect]\label{thm:Zeno}
Let $U_\kick$ be the unitary defined in \cref{eq:Ukickdef} with $m$ distinct eigenvalues, inverse spectral gap $\xi$, and a set of orthogonal projectors $\{P_\mu\}$. Let $G_\Zeno = \sum_{\mu} P_\mu G P_\mu$ denote the projection of a Hamiltonian $G$ onto the subspaces defined by $\{P_\mu\}$.
We have 
\begin{align}
	\epsilon_{\text{Zeno}} \equiv &\norm{U_{\text{kick}}^{\dagger r}\left(U_{\text{kick}}e^{-i\frac{t}{r}G}\right)^r - e^{-itG_{\text{Zeno}}}}
	\leq\frac{2\xi\sqrt{m}\norm{G}^2t^2\log r}{r}+\frac{\xi\sqrt{m}\norm{G}t}{r}
	\leq\frac{3\xi\sqrt{m}\norm{G}^2t^2\log r}{r}.\label{eq:new_zeno_bound}
\end{align}
\end{theorem}

To prove \cref{thm:Zeno}, we rewrite the evolution as
\begin{equation}
	U_{\text{kick}}^{\dagger r}\left(e^{-i\frac{t}{r}G}U_{\text{kick}}\right)^r
	=e^{-i\frac{t}{r}G_{r}}e^{-i\frac{t}{r}G_{r-1}}\cdots e^{-i\frac{t}{r}G_1},
\end{equation}
where we have defined
\begin{equation}
	G_k\equiv U_{\text{kick}}^{\dagger k}GU_{\text{kick}}^{k}.
\end{equation}
Letting $G_{[1,r]}\equiv G_1+\cdots+G_{r}$, the first step of our proof is to establish the error bound
\begin{equation}
\label{eq:destructive_trotter_error}
	\norm{e^{-i\frac{t}{r}G_{r}}\cdots e^{-i\frac{t}{r}G_1}-e^{-i\frac{t}{r}G_{[1,r]}}}\leq\frac{2\xi\sqrt{m}\norm{G}^2t^2\log r}{r}.
\end{equation}
This is the spectral-norm error of the first-order Trotter formula \cite{childsTheoryTrotterError2020}. However, a naive error analysis in terms of the commutators between $G_j$ (see \cite[Proposition 15]{childsTheoryTrotterError2020} for example) gives a bound that does not decrease with $r$ and thus fails to establish the desirable bound. Instead, we seek a better analysis that exploits the spectral information of $U_{\text{kick}}$ \cite{burgarthGeneralizedProductFormulas2019}.

The starting point of our analysis is the established von Neumann's ergodic theorem whose proof is included for completeness.
\begin{theorem}[Von Neumann's ergodic theorem]
	Let $U$ be a unitary operator and $U=\sum_{\mu=1}^{m}e^{-i\phi_\mu}P_\mu$ be its spectral decomposition, with $\phi_1=0$ and $\phi_\mu$ distinct. Then,
	\begin{equation}
	\norm{\frac{1}{r}\sum_{k=1}^{r}U^k-P_1}\leq\frac{\xi_1}{r},
	\end{equation}
	where
	\begin{equation}
		\xi_1:=2\max_{\nu\neq 1}\abs{e^{-i\phi_\nu}-1}^{-1}
		=\max_{\nu\neq 1}\abs{\sin\left(\frac{\phi_\nu}{2}\right)}^{-1}.
	\end{equation}
\end{theorem}
\begin{proof}
	The bound follows from
	\begin{equation}
	\begin{aligned}
		\norm{\frac{1}{r}\sum_{k=1}^{r}U^k-P_1}
		&=\norm{\left(\frac{1}{r}\sum_{k=1}^{r}U^k-P_1\right)\sum_{\nu=1}^{m}P_\nu}\\
		&=\norm{\frac{1}{r}\sum_{\nu=1}^{m}\sum_{k=1}^{r}e^{-ik\phi_\nu}P_\nu-P_1}\\
		&=\norm{\frac{1}{r}\sum_{\nu\neq 1}e^{-i\phi_\nu}\frac{1-e^{-ir\phi_\nu}}{1-e^{-i\phi_\nu}}P_\nu}\\
		&=\frac{1}{r}\max_{\nu\neq 1}\abs{\frac{1-e^{-ir\phi_\nu}}{1-e^{-i\phi_\nu}}}
		\leq\frac{\xi_1}{r}.
	\end{aligned}
	\end{equation}
\end{proof}
We note that the condition $\phi_1 = 0$ is not restrictive as we can always make $\phi_1 = 0$ by adding a global phase to $U_\kick$ \cite{Facchi04}.
\begin{corollary}
	\label{cor:zeno_subspace}
	Let $U$ be a unitary operator and $U=\sum_{\mu=1}^{m}e^{-i\phi_\mu}P_\mu$ be its spectral decomposition. Then, for any operator $G$,
	\begin{equation}
	\norm{\frac{1}{r}\sum_{k=1}^{r}U^kGU^{-k}-\sum_{\mu=1}^mP_\mu GP_\mu}\leq\frac{\xi\sqrt{m}\norm{G}}{r},
	\end{equation}
	where
	\begin{equation}
	\xi:=2\max_{\mu\neq \nu}\abs{e^{-i\phi_\mu}-e^{-i\phi_\nu}}^{-1}
	=\max_{\mu\neq\nu}\abs{\sin\left(\frac{\phi_\nu-\phi_\mu}{2}\right)}^{-1}.
	\end{equation}
\end{corollary}
\begin{proof}
	The claimed bound follows from
	\begin{equation}
	\begin{aligned}
		\norm{\frac{1}{r}\sum_{k=1}^{r}U^kGU^{-k}-\sum_{\mu=1}^mP_\mu GP_\mu}
		&=\norm{\left(\frac{1}{r}\sum_{k=1}^{r}U^kGU^{-k}-\sum_{\mu=1}^mP_\mu GP_\mu\right)\sum_{\nu=1}^{m}P_\nu}\\
		&=\norm{\sum_{\nu=1}^{m}\frac{1}{r}\sum_{k=1}^{r}\left(e^{i\phi_\nu}U\right)^kGP_\nu-\sum_{\nu=1}^mP_\nu GP_\nu}\\
		&\leq\sqrt{\sum_{\nu=1}^{m}\norm{\frac{1}{r}\sum_{k=1}^{r}\left(e^{i\phi_\nu}U\right)^kGP_\nu-P_\nu GP_\nu}^2}\\
		&\leq\sqrt{m}\norm{G}\max_{\nu}\norm{\frac{1}{r}\sum_{k=1}^{r}\left(e^{i\phi_\nu}U\right)^k-P_\nu}
		\leq\frac{\xi\sqrt{m}\norm{G}}{r},
	\end{aligned}
	\end{equation}
	where the first inequality follows from the bound
	\begin{equation}
	\begin{aligned}
	    \norm{\sum_{\nu=1}^{m}A_{\nu} P_\nu}
	    &=\sqrt{\norm{\left(\sum_{\nu=1}^{m}A_{\nu} P_\nu\right)\left(\sum_{\mu=1}^{m}A_{\mu} P_\mu\right)^\dagger}}\\
	    &=\sqrt{\norm{\sum_{\nu=1}^{m}A_{\nu} P_\nu A_\nu^\dagger}}
	    \leq\sqrt{\sum_{\nu=1}^{m}\norm{A_{\nu} P_\nu A_\nu^\dagger}}
	    =\sqrt{\sum_{\nu=1}^{m}\norm{A_{\nu} P_\nu}^2}.
	\end{aligned}
	\end{equation}
\end{proof}

As aforementioned, a naive analysis of the Trotter error fails to provide the desirable bound for quantum Zeno effect. Instead, we use a recursive approach to estimate the Trotter error \cref{eq:destructive_trotter_error}. 
\begin{lemma}\label{lem:TrotterZeno}
	Define $G_{[k_0,k_1]}\equiv\sum_{k=k_0}^{k_1}G_k$ for $k_0\leq k_1$. 
	For any $s \geq 1$ and $\dt$, we have
	\begin{align} 
		\norm{\prod_{k=1}^s e^{-iG_k \dt} - e^{-iG_{[1,s]}\dt}} \leq 2\xi\sqrt{m} \norm{G}^2\dt^2s\log_2 s.	 
	\end{align}	
\end{lemma}
Note that at $s = r$ and $\dt = t/r$, \cref{lem:TrotterZeno} reduces to \cref{eq:destructive_trotter_error}.
We prove \cref{lem:TrotterZeno} by induction on $s$.
Suppose \cref{lem:TrotterZeno} holds for $s = s_1$ and $s = s_2$ such that $\abs{s_2-s_1}\leq 1$, we shall prove that it holds for $s = s_1+s_2$.
Using the triangle inequality
\begin{align} 
	 \norm{\prod_{k=1}^{s_1+s_2} e^{-iG_k \dt} - e^{-iG_{[1,s_1+s_2]}\dt}}
	 &\leq \norm{\prod_{k=1}^{s_1} e^{-iG_k \dt} - e^{-iG_{[1,s_1]}\dt}}
	 + \norm{\prod_{k=s_1+1}^{s_1+s_2} e^{-iG_k \dt} - e^{-iG_{[s_1+1,s_1+s_2]}\dt}}\nonumber\\
	 &\qquad+ \norm{e^{-iG_{[1,s_1+s_2]}\dt} - e^{-iG_{[s_1+1,s_1+s_2]}\dt}e^{-iG_{[1,s_1]}\dt}}\\
	 &\leq 2\xi\sqrt{m} \norm{G}^2\dt^2(s_1\log_2s_1 + s_2\log_2 s_2)
	 + \frac12\norm{\comm{G_{[s_1+1,s_1+s_2]},G_{[1,s_1]}}}\dt^2,\label{eq:43274239}
\end{align}
where we have used the inductive hypothesis and the Trotter error bound \cite[Eq.\ (143)]{childsTheoryTrotterError2020} in the last inequality. To bound the commutator norm, we use the following lemma.
\begin{lemma}\label{lem:commG}
For any $k_0\leq k_1, j_0\leq j_1$, we have
	\begin{align}
		&\norm{\left[G_{[k_0, k_1]},G_{[j_0,j_1]}\right]}
		\leq 2\left(j_1-j_0+k_1-k_0+2\right)\xi\sqrt{m}\norm{G}^2.
	\end{align}
\end{lemma}
\begin{proof}
	We have
	\begin{equation}
	\begin{aligned}
		\norm{\left[G_{k_0\leq k\leq k_1},G_{j_0\leq j\leq j_1}\right]}
		&=\norm{\left[\sum_{k=k_0}^{k_1}G_k,\sum_{j=j_0}^{j_1}G_j\right]}\\
		&\leq\norm{\left[\sum_{k=k_0}^{k_1}G_k,\sum_{j=j_0}^{j_1}G_j-(j_1-j_0+1)\sum_{\mu=1}^mP_\mu GP_\mu\right]}
		\nonumber\\&
		\quad+\norm{\left[\sum_{k=k_0}^{k_1}G_k-(k_1-k_0+1)\sum_{\mu=1}^mP_\mu GP_\mu,(j_1-j_0+1)\sum_{\mu=1}^mP_\mu GP_\mu\right]}\\
		&\leq 2(k_1-k_0+1)\norm{G}\left(\xi\sqrt{m}\norm{G}\right)+2\left(\xi\sqrt{m}\norm{G}\right)(j_1-j_0+1)\norm{G}\\
		&=2(j_1+k_1-j_0-k_0+2)\xi\sqrt{m}\norm{G}^2,
	\end{aligned}
	\end{equation}
	where we have used \cref{cor:zeno_subspace} to prove the second inequality. Therefore, the lemma follows.
\end{proof}
Applying \cref{lem:commG} to \cref{eq:43274239}, we have
\begin{align} 
	&\norm{\prod_{k=1}^{s_1+s_2} e^{-iG_k \dt} - e^{-iG_{[1,s_1+s_2]}\dt}} 
	\leq (2s_1\log_2 s_1 + 2s_2\log_2 s_2 +s_1+s_2)\xi\sqrt{m} \norm{G}^2\dt^2. \label{eq:zenorecursive}
\end{align}
We now apply the above equation repeatedly to prove \cref{lem:TrotterZeno}.
Note that \cref{lem:TrotterZeno} holds trivially for $s = 1$.
Suppose that it holds for all $s\leq s_0$ for some $s_0 \geq 1$. We shall prove that it holds for $s = s_0+1$.

First, we consider the case where $s$ is even, i.e. there exists an integer $l\geq 1$ such that $s = 2l$. 
Applying \cref{eq:zenorecursive} with $s_1=s_2=l$, we get
\begin{align} 
	 \norm{\prod_{k=1}^{s} e^{-iG_k \dt}-e^{-iG_{[1,s]\dt}}}
	 &\leq (2l\log_2 l + 2l\log_2 l +l+l)\xi\sqrt{m} \norm{G}^2\dt^2\\
	 &= (2s\log_2(s/2) + s)\xi\sqrt{m} \norm{G}^2\dt^2\\
	 &< 2s\log_2 s\ \xi\sqrt{m} \norm{G}^2\dt^2.
\end{align}
Therefore, \cref{lem:TrotterZeno} holds if $s$ is even.

When $s$ is odd, there exists an integer $l\geq 1$ such that $s = 2l+1$.
Applying \cref{lem:TrotterZeno} with $s_1 = l$ and $s_2 = l+1$, we have
\begin{align} 
	\norm{\prod_{k=1}^{s} e^{-iG_k \dt}-e^{-iG_{[1,s]\dt}}}
	 &\leq (2l\log_2 l + 2(l+1)\log_2 (l+1) +2l+1)\xi\sqrt{m} \norm{G}^2\dt^2. \label{eq:zenofdskfhdksk}
\end{align}
Let 
\begin{align} 
	g(x) =  2x\log_2 x + 2(x+1)\log_2 (x+1) +2x+1 - 2(2x+1)\log_2(2x+1).
\end{align}
It is straightforward to verify that $g(1) < 0$ and 
\begin{align} 
	g'(x) =  2 \log_2 \frac{2x(1+x)}{(1+2x)^2} < 0
\end{align}
for all $x \geq 1$.
Therefore, $g(x) \leq 0$ for all $x \geq 1$.  
Applying this bound to \cref{eq:zenofdskfhdksk}, we get
\begin{align} 
	\norm{\prod_{k=1}^{s} e^{-iG_k \dt}-e^{-iG_{[1,s]\dt}}}
	 &\leq 2(2l+1) \log_2 (2l+1)\xi\sqrt{m} \norm{G}^2\dt^2\\
	 &= 2s\log_2 s\ \xi\sqrt{m} \norm{G}^2\dt^2.
\end{align}
Thus, \cref{lem:TrotterZeno} holds for odd $s$ too.
By induction, \cref{lem:TrotterZeno} holds for all $s\geq 1$.

Combining \cref{lem:TrotterZeno} with
\begin{equation}
\label{eq:zeno_converge}
	\norm{e^{-i\frac{t}{r}G_{[1,r]}}-e^{-itG_{\text{Zeno}}}} \leq \frac{t}{r} \norm{G_{[1,r]} - rG_{\text{Zeno}}}
	\leq\frac{\xi\sqrt{m}\norm{G}t}{r}
\end{equation}
from \cref{cor:zeno_subspace},
we prove \cref{eq:new_zeno_bound}.

\section{Symmetry protection by quantum Zeno effect}
\label{sec:ZenoxSP}
In this section, we make a formal connection between the symmetry protection technique and the quantum Zeno effect.
In particular, we show how the quantum Zeno framework provides an alternative explanation for the suppression of simulation error under symmetry protection. 

We first note that the symmetry transformations in our scheme are analogous to the kicks in the quantum Zeno framework.
Suppose that the symmetry transformations have the form $C_k = C_0^k$, where $C_0 \in \mathcal S$ is also a symmetry transformation. 
Let
\begin{align} 
 	 C_0 = \sum_\mu e^{-i\phi_\mu}P_\mu
\end{align} 
be the spectral decomposition of $C_0$, with $e^{-i\phi_\mu}$ being the distinct eigenvalues and $P_\mu$ being the projectors onto the respective eigensubspaces. The condition on $e^{-i\phi_\mu}$ being distinct ensures that $C_0$ satisfies the definition of $U_\kick$ in \cref{eq:Ukickdef}.

With $e^{-iH\dt}$ being approximated by a circuit $S_\dt$ in each time step, our symmetry-protected simulation becomes
\begin{align} 
	\prod_{k=1}^r C_k^\dag S_\dt C_k = (C_0^{\dag})^r (e^{-iH_\eff \dt}C_0)^r, \label{eq:ZenoSP}
\end{align}
where $H_\eff$ is the generator of $S_\dt$ and exists for a small enough $\dt$ (see \cref{lem:Htilde}).
Comparing \cref{eq:ZenoSP} with \cref{eq:Zenodynamics}, we identify $C_0 = U_\kick$.
Therefore, by \cref{thm:Zeno}, the symmetry protected simulation is effectively described by
\begin{align} 
	\prod_{k=1}^r C_k^\dag S_\dt C_k\rightarrow e^{-iH_{\eff,\Zeno} t}, \label{eq:Zenolargerlimit}
\end{align}
in the large $r$ limit, where $H_{\eff,\Zeno} = \sum_{\mu} P_\mu H_\eff P_\mu$.

Recall that $H_\eff$ is the effective Hamiltonian corresponding the Trotterized evolution $S_\dt$.
For small $\dt$, it is a sum of the true Hamiltonian $H$ that we are simulating and a small error term $V$ (due to the use of Trotterization):
\begin{align} 
	H_\eff = H + V. 
\end{align}
Therefore, under the symmetry protection, the effective Hamiltonian is replaced by its projection onto the Zeno subspaces:
\begin{align} 
	H_\eff \rightarrow H_{\eff,\Zeno} = H + V_\Zeno, 
\end{align}
where $V_\Zeno = \sum_\mu P_\mu V P_\mu$ is the corresponding projection of $V$.
In particular, if the error $V$ does not respect the symmetry, the projection $V_\Zeno$ could be much smaller than the error $V$ in an unprotected simulation.
The quantum Zeno framework therefore provides alternative intuition for the error suppression from the symmetry protection.
We note, however, that choosing the symmetry transformations $C_k$ independently, instead of $C_k = C_0^k$ considered in this section, could lead to more reduction of the simulation error, and we demonstrate this advantage in \cref{sec:app}.

We make these arguments rigorous by proving a bound analogous to that in \cref{thm:Zeno} for symmetry-protected quantum simulation. Specifically, we consider $G = H_\eff = H+V$, where $\comm{H,U_\kick} =0$.
Note that under this assumption, the distinctiveness of the eigenvalues of $U_\kick$ ensures that $\comm{P_\mu,H} =0$ for all $\mu$ in the spectral decomposition of $U_\kick$.
We will also denote by $V_k = U_{\text{kick}}^{\dagger k}VU_{\text{kick}}^{k} = G_k - H$.

\begin{theorem}[Symmetry protection by quantum Zeno effect]\label{thm:Zeno2}
Let $U_\kick$ be the unitary defined in \cref{eq:Ukickdef} and suppose that $G = H + V$ such that $\comm{H,U_\kick} = 0$.
Let $G_\Zeno = \sum_{\mu} P_\mu G P_\mu  = H + \sum_{\mu} P_\mu V P_\mu$ denote the projection of $G$ onto the subspaces defined by a set of orthogonal projectors $\{P_\mu\}$ in the spectral decomposition of $U_\kick$.
We have 
\begin{align}
	\epsilon_{\text{Zeno}} \equiv \norm{U_{\text{kick}}^{\dagger r}\left(U_{\text{kick}}e^{-i\frac{t}{r}G}\right)^r - e^{-itG_{\text{Zeno}}}}
	\leq\frac{2\xi\sqrt{m}\norm G\norm{V}t^2\log r}{r}+\frac{\xi\sqrt{m}\norm{V}t}{r}
	&\leq\frac{3\xi\sqrt{m}\norm G\norm{V}t^2\log r}{r},\label{eq:new_zeno_bound_2}
\end{align}
where $\xi$ is the inverse spectral gap defined in \cref{eq:CZeno_def}.
\end{theorem}

Note that this bound is stronger than \cref{eq:new_zeno_bound} in that the dependence on the norm of the Hamiltonian is improved from $\norm{G}^2$ to $\norm{G}\norm{V}$.
To prove \cref{eq:new_zeno_bound_2}, we derive a different version of \cref{lem:TrotterZeno} for the case $G = H+V$, where $\comm{H,U_\kick} =0$.
\begin{lemma}\label{lem:TrotterZeno2}
	Suppose $G = H+V$, where $\comm{H,U_\kick} =0$.
	For any $s \geq 1$ and $\dt$, we have
	\begin{align} 
		\norm{\prod_{k=1}^s e^{-iG_k \dt} - e^{-iG_{[1,s]}\dt}} \leq 2\xi\sqrt{m} \norm{G}\norm{V}\dt^2s\log_2 s.	 
	\end{align}	
\end{lemma}
Again, we prove \cref{lem:TrotterZeno2} by induction on $s$.
Suppose \cref{lem:TrotterZeno2} holds for $s = s_1$ and $s = s_2$ such that $\abs{s_2-s_1}\leq 1$, we shall prove that it holds for $s = s_1+s_2$.
Using the triangle inequality
\begin{align} 
	 \norm{\prod_{k=1}^{s_1+s_2} e^{-iG_k \dt} - e^{-iG_{[1,s_1+s_2]}\dt}}
	 &\leq \norm{\prod_{k=1}^{s_1} e^{-iG_k \dt} - e^{-iG_{[1,s_1]}\dt}}
	 + \norm{\prod_{k=s_1+1}^{s_1+s_2} e^{-iG_k \dt} - e^{-iG_{[s_1+1,s_1+s_2]}\dt}}\nonumber\\
	 &\qquad+ \norm{e^{-iG_{[1,s_1+s_2]}\dt} - e^{-iG_{[s_1+1,s_1+s_2]}\dt}e^{-iG_{[1,s_1]}\dt}}\\
	 &\leq 2\xi\sqrt{m} \norm{G}\norm V\dt^2(s_1\log_2s_1 + s_2\log_2 s_2)
	 + \frac12\norm{\comm{G_{[s_1+1,s_1+s_2]},G_{[1,s_1]}}}\dt^2.\label{eq:4327424324339}
\end{align}
To bound the commutator norm, we use a modified version of \cref{lem:commG}.
\begin{lemma}\label{lem:commG2}
Given $G = H + V$ and $\comm{H,U_\kick} = 0$, we have
	\begin{align}
		&\norm{\left[G_{[k_0, k_1]},G_{[j_0,j_1]}\right]}
		\leq 2\left(j_1-j_0+k_1-k_0+2\right)\xi\sqrt{m}\norm{G}\norm V.
	\end{align}
\end{lemma}
\begin{proof}
We have
	\begin{equation}
	\begin{aligned}
		&\norm{\left[G_{k_0\leq k\leq k_1},G_{j_0\leq j\leq j_1}\right]}
		=\norm{\left[\sum_{k=k_0}^{k_1}G_k,\sum_{j=j_0}^{j_1}G_j\right]}\\
		&\leq\norm{\left[\sum_{k=k_0}^{k_1}G_k,\sum_{j=j_0}^{j_1}G_j-(j_1-j_0+1)\sum_{\mu=1}^mP_\mu GP_\mu\right]}
		+\norm{\left[\sum_{k=k_0}^{k_1}G_k-(k_1-k_0+1)\sum_{\mu=1}^mP_\mu GP_\mu,(j_1-j_0+1)\sum_{\mu=1}^mP_\mu GP_\mu\right]}\\
		&=\norm{\left[\sum_{k=k_0}^{k_1}G_k,\sum_{j=j_0}^{j_1}V_j-(j_1-j_0+1)\sum_{\mu=1}^mP_\mu VP_\mu\right]}
		 +\norm{\left[\sum_{k=k_0}^{k_1}V_k-(k_1-k_0+1)\sum_{\mu=1}^mP_\mu VP_\mu,(j_1-j_0+1)\sum_{\mu=1}^mP_\mu GP_\mu\right]}\\
		&\leq 2(k_1-k_0+1)\norm{G}\left(\xi\sqrt{m}\norm{V}\right)+2\left(\xi\sqrt{m}\norm{V}\right)(j_1-j_0+1)\norm{G}\\
		&=2(j_1-j_0+k_1-k_0+2)\xi\sqrt{m}\norm G\norm{V},
	\end{aligned}
	\end{equation}
	where $V_k = U_{\text{kick}}^{\dagger k}VU_{\text{kick}}^{k} = G_k - H$ as mentioned above.
	Therefore, the lemma follows.
\end{proof}
Applying \cref{lem:commG2} to \cref{eq:4327424324339}, we have
\begin{align} 
	&\norm{\prod_{k=1}^{s_1+s_2} e^{-iG_k \dt} - e^{-iG_{[1,s_1+s_2]}\dt}} 
	\leq {(2s_1\log_2 s_1 + 2s_2\log_2 s_2 +s_1+s_2)}\xi\sqrt{m} \norm{G}\norm V\dt^2. 
\end{align}
Using this bound and an inductive argument similar to the proof of \cref{lem:TrotterZeno}, we prove \cref{lem:TrotterZeno2}.
Finally, combining \cref{lem:TrotterZeno2} at $s = r$ with 
\begin{align} 
	 \norm{e^{-i\frac{t}{r}G_{[1,r]}}-e^{-itG_{\text{Zeno}}}}
	\leq \norm{G_{[1,r]}-rG_\Zeno}\frac{t}{r}
	\leq \norm{\sum_{k=1}^r V_k-rV_\Zeno}\frac{t}{r}
	\leq\frac{\xi\sqrt{m}\norm{V}t}{r},
\end{align}
we obtain \cref{eq:new_zeno_bound_2}.

\section{A general bound on the Trotter error}\label{sec:pf1}

In \cref{sec:theory}, we prove a bound on the simulation error under the protection from a special class of symmetry transformations $C_k = C_0^k$. 
In this section, we prove a similar, but more general, bound without making such an assumption.

Given a fixed total evolution time $t$, we first estimate the number of Trotter steps $r$ required to simulate $\exp(-iHt)$ so that the total additive error of the simulation meets a threshold~$\epsilon$. 
Suppose the Hamiltonian $H = \sum_{\mu=1}^L {H_\mu}$ is a sum of $L$ Hamiltonian terms $H_\mu$ such that each $e^{-iH_\mu \dt}$ can be readily simulated on quantum computers.
Again, we define the following quantities
\begin{align} 
	&\gamma \equiv  \sum_{\mu=1}^L \sum_{\nu = \mu+1}^L\norm{\comm{H,\comm{H_\nu,H_\mu}}},\\
	&\beta \equiv  \sum_{\mu=1}^L \sum_{\nu = \mu+1}^L \sum_{\nu' = \nu}^{L}\norm{\comm{H_{\nu'},\comm{H_\nu,H_\mu}}},\\
	&\alpha \equiv \sum_{\mu = 1}^{L} \sum_{\nu = \mu+1}^L \norm{\comm{H_\nu,H_\mu}},
\end{align}
that are independent of $t,r$.

The first-order Trotterization approximates $\exp(-i H \dt)$ by
\begin{align} 
	 S_\dt = \prod_{\mu = 1}^L e^{-iH_\mu\delta t},\label{eq:PF1}
\end{align}
where $\prod_{\mu = 1}^L U_\mu \equiv U_L\dots U_2.U_1$ is an ordered product.

To get an accurate scaling of the gate count with the error tolerance, time, and the system size, we extend the approach in Ref.~\cite{childsTheoryTrotterError2020} to estimate the higher-order contributions to the total error.
First, we estimate the higher-order contributions to the additive error in one Trotter step.
\begin{lemma}\label{lem:PF1error}
Assuming $\beta \dt \leq 2\alpha$ and $\alpha^2\dt \leq \gamma + \beta$, the Trotter error in approximating $U_\dt = \exp(-iH\dt)$ by $S_\dt$ in \cref{eq:PF1} is given by
\begin{align} 
	\mathcal E_\dt \equiv U_\dt - S_\dt = U_\dt v_0 \frac{\dt^2}{2} + \widetilde{\mathcal V}(\delta t), 
\end{align}
where $v_0$ is defined in \cref{eq:mathcalE_0} and $\widetilde{\mathcal V}(\delta t)$ is an operator bounded by
\begin{align}
	\norm{\widetilde{\mathcal V}(\delta t)}
	\leq \Lambda \dt^3,
\end{align}
with $\Lambda = \frac56(\gamma+\beta)$.
\end{lemma}
\begin{proof}
From \cite[Theorem 8]{childsTheoryTrotterError2020}, we have
\begin{align} 
	S_\dt &= e^{-iH\dt} 
	\Texp{-i\int_0^\dt d\tau_1 \widetilde{{F}}(\tau_1) },
\end{align}
where $\Texp{}$ is the time-ordered exponential,
\begin{align} 
	\widetilde{F}(\tau_1) = e^{i\tau_1 \ad_H} \sum_{\mu = 1}^{L}\left(\prod_{\nu = \mu+1}^L {e^{-i\tau_1 \ad_{H_\nu}}} H_\mu - H_\mu\right), 
\end{align}
$\ad_A B \equiv [A,B]$, and $e^{-it\ad_A}B = e^{-itA}Be^{itA}$.
Note that the summand in the definition of $\widetilde{F}(\tau_1)$ is of order $\O{\tau_1}$.
Therefore, we can rewrite it as (See \cite[Theorem 10]{childsTheoryTrotterError2020} or use a direct differentiation):
\begin{align} 
	 &\prod_{\nu = \mu+1}^L {e^{-i\tau_1 \ad_{H_\nu}}} H_\mu - H_\mu
	 = -i\sum_{\nu = \mu+1}^L \int_0^{\tau_1}d\tau_2 \prod_{\nu' = \nu+1}^{L}{e^{-i\tau_1 \ad_{H_{\nu'}}}}
	 e^{-i\tau_2 \ad_{H_\nu}}
	 \comm{H_\nu,H_\mu}\\
	 &=-i\sum_{\nu=\mu+1}^L \comm{H_\nu,H_\mu}\tau_1 -i  \sum_{\nu = \mu+1}^L \int_0^{\tau_1}d\tau_2 
	 \underbrace{
	 \left(\prod_{\nu' = \nu+1}^{L}{e^{-i\tau_1 \ad_{H_{\nu'}}}}
	 e^{-i\tau_2 \ad_{H_\nu}}
	 \comm{H_\nu,H_\mu} - \comm{H_\nu,H_\mu}\right)}_{\equiv G_{\mu,\nu}(\tau_1,\tau_2)} .
\end{align}
Again, we note that $G(\tau_1) = \O{\tau_1 + \tau_2}$.
Therefore, we can rewrite it (using either \cite[]{childsTheoryTrotterError2020} or a direct differentiation) as
\begin{align} 
 	G_{\mu,\nu}(\tau_1,\tau_2) = -i \sum_{\nu' = \nu+1}^L \prod_{s = \nu'+1}^L e^{-i\tau_1\ad_{H_{s}}}
 	\int_0^{\tau_1} d\tau_3 e^{-i\tau_3\ad{H_{\nu'}}}\comm{H_{\nu'},\comm{H_\nu,H_\mu}}\nonumber\\
 	-i\frac{\tau_2}{\tau_1}\prod_{\nu' = \nu+1}^L e^{-i\tau_1\ad_{H_{\nu'}}}
 	\int_0^{\tau_1} d\tau_3 e^{-i\tau_3\frac{\tau_2}{\tau_1}\ad_{H_{\nu}}}\comm{H_{\nu},\comm{H_\nu,H_\mu}}.
\end{align} 
Using the triangle inequality, we have
\begin{align} 
	\norm{G_{\mu,\nu}(\tau_1,\tau_2)} \leq \tau_1 \sum_{\nu'=\nu}^L \norm{\comm{H_{\nu'},\comm{H_\nu,H_\mu}}}.
\end{align}
Therefore,
\begin{align} 
	\norm{\widetilde{F}(\tau_1)} &\leq 
	\tau_1  \sum_{\mu = 1}^{L} \sum_{\nu = \mu+1}^L \norm{\comm{H_\nu,H_\mu}} 
	+\frac{\tau_1^2}{2} \sum_{\mu = 1}^{L} \sum_{\nu = \mu+1}^L\sum_{\nu' = \nu}^L \norm{\comm{H_{\nu'},\comm{H_\nu,H_\mu}}} 
\end{align}
In addition, we have
\begin{align} 
	&e^{i\tau_1 \ad_H} \comm{H_\nu,H_\mu} - \comm{H_\nu,H_\mu}
	= i\int_0^{\tau_1}d\tau_2 e^{i\tau_2 \ad_H}\comm{H,\comm{H_\nu,H_\mu}}. 
\end{align}
Therefore, with $v_0 = \sum_{\mu = 1}^L\sum_{\nu = \mu+1}^L \comm{H_\nu,H_\mu}$, we have
\begin{align} 
	\widetilde{F}(\tau_1) + iv_0 \tau_1
	&= -i \sum_{\mu=1}^L \sum_{\nu = \mu+1}^L
	\bigg( e^{i\tau_1 \ad_H}\comm{H_\nu,H_\mu}\tau_1- \comm{H_\nu,H_\mu}\tau_1
	+ e^{i\tau_1 \ad_H} \int_0^{\tau_1} d\tau_2 G_{\mu,\nu}(\tau_1,\tau_2)\bigg)\nonumber\\
	&= -i \sum_{\mu=1}^L \sum_{\nu = \mu+1}^L \int_0^{\tau_1}d\tau_2
	\bigg( i\tau_1  e^{i\tau_2 \ad_H}\comm{H,\comm{H_\nu,H_\mu}}
	+ e^{i\tau_1 H}  G_{\mu,\nu}(\tau_1,\tau_2)\bigg)
\end{align}
Expanding the time-ordered exponential, we have
\begin{align} 
	\Texp{-i\int_0^\dt d\tau_1 \widetilde{F}(\tau_1)} 
	&= \mathbb I -i \int_0^\dt d\tau_1 \widetilde{F}(\tau_1) 
	- \int_0^\dt d\tau_1 \int_0^{\tau_1} d\tau_2 \widetilde{F}(\tau_1)\widetilde{F}(\tau_2)\Texp{-i\int_0^{\tau_2} d\tau_3 \widetilde{F}(\tau_3)}\\
	&= \mathbb I - \frac{\dt^2}{2}v_0 -i \int_0^\dt d\tau_1 [\widetilde{F}(\tau_1) + iv_0\tau_1] 
	\nonumber\\&
	\quad- \int_0^\dt d\tau_1 \int_0^{\tau_1} d\tau_2 \widetilde{F}(\tau_1)\widetilde{F}(\tau_2)\Texp{-i\int_0^{\tau_2} d\tau_3 \widetilde{F}(\tau_3)}. 
\end{align}
Therefore, we have
\begin{align} 
	\norm{S_\dt - e^{-iH\dt} + e^{-iH\dt} v_0\frac{\dt^2}{2}}
	&\leq \int_0^\dt d\tau_1 \norm{\widetilde{F}(\tau_1)+iv_0 \tau_1} 
	+ \int_0^\dt d\tau_1 \int_0^{\tau_1} d\tau_2 \norm{\widetilde{F}(\tau_1)}\norm{\widetilde{F}(\tau_2)} 
	\\
	&\leq \frac{\gamma+\beta}{3}  \dt^3 
	+ \frac{\dt^4}{8} \left(\alpha  +\frac{\dt}{2}\beta \right)^2.
\end{align}
In particular, assuming $\beta \dt \leq 2\alpha$ and $\alpha^2\dt \leq \gamma + \beta$, we have
\begin{align} 
	\norm{\mathcal E_\dt - U_\dt v_0 \frac{\dt^2}{2}} \leq \Lambda \dt^3,
\end{align}
with $\Lambda = \frac56 (\gamma + \beta)$.
Therefore, \cref{lem:PF1error} follows.
\end{proof}
As a result of \cref{lem:PF1error}, we can bound the additive error in one Trotter step:
\begin{align}
	\norm{\mathcal E_\dt} \leq \frac{\norm{v_0}}{2}\dt^2 + \Lambda \dt^3.
\end{align}
Therefore, we arrive at a bound for the total error for the simulation 
\begin{align}
	\epsilon 
	&=\norm{U_t - \prod_{k=1}^r C_k^\dag S_{\dt}C_k} \\
	&\leq\norm{\sum_{k=1}^r C_k^\dag U_{k\dt}^\dag \mathcal E_\dt U_{k\dt} C_k} + \sum_{j=2}^{r} \binom{r}{j} \norm{\mathcal E_\dt}^j\\
	&\leq\norm{\sum_{k=1}^r C_k^\dag U_{k\dt}^\dag v_0 U_{k\dt}C_k}\frac{\dt^2}{2}+r\Lambda\dt^3 + \sum_{j=2}^{r} (r\norm{\mathcal E_\dt})^j\\
	&\leq\norm{\sum_{k=1}^r C_k^\dag U_{k\dt}^\dag  v_0 U_{k\dt}C_k}\frac{\dt^2}{2}+r\Lambda\dt^3 + 2r^2\norm{\mathcal E_\dt}^2\\
	&\leq\bigg\Vert\underbrace{\frac1r\sum_{k=1}^r C_k^\dag U_{k\dt}^\dag v_0 U_{k\dt}C_k}_{\equiv \overbar{ v_0}}\bigg\Vert\frac{t^2}{2r}+\Lambda \frac{t^3}{r^2} 
	+ 2\left(\frac{1}{2}\norm{ v_0} + \Lambda \frac{t}{r}\right)^2\frac{t^4}{r^2},
\end{align}
where $U_{k\dt} = \exp(-iHk\dt)$ and we have assume $r\norm{\mathcal E_\dt} \leq 1/2$ to bound the sum over $j$.
This bound again has the same feature as the bound in \cref{lem:pf1finalbound}:
the total error, to the lowest-order, scales with $\norm{\overbar{ v_0}}$---an averaged version of $ v_0$ under the symmetry transformations---instead of scaling with $\norm{ v_0}$.
Note, however, that the definition of $\overbar{ v_0}$ here, with the addition of the transformations under $U_{k\dt}$, is slightly different from \cref{lem:pf1finalbound}.

\section{Proof of \cref{lem:Htilde}}\label{sec:PF1exp_error_proof}
In this section, we prove \cref{lem:Htilde}, which provides a bound on the error in one Trotter step.
\begin{proof}
From \cite[Theorem 8]{childsTheoryTrotterError2020}, we have
\begin{align} 
	S_\dt &=  
	\Texp{-i\int_0^\dt d\tau_1 \left(H + F(\tau_1)\right) },\label{eq:experrorapp}
\end{align}
where $\Texp{}$ is the time-ordered exponential,
\begin{align} 
	F(\tau_1) = \sum_{\mu = 1}^{L}\left(\prod_{\nu = \mu+1}^L {e^{-i\tau_1 \ad_{H_\nu}}} H_\mu - H_\mu\right), 
\end{align}
$\ad_A B \equiv [A,B]$, and $e^{-it\ad_A}B = e^{-itA}Be^{itA}$.
Note that the summand in the definition of $F(\tau_1)$ is of order $\O{\tau_1}$.
Therefore, we can rewrite it as (See \cite[Theorem 10]{childsTheoryTrotterError2020} or use a direct differentiation):
\begin{align} 
	 &\prod_{\nu = \mu+1}^L {e^{-i\tau_1 \ad_{H_\nu}}} H_\mu - H_\mu
	 = -i\sum_{\nu = \mu+1}^L \int_0^{\tau_1}d\tau_2 \prod_{\nu' = \nu+1}^{L}{e^{-i\tau_1 \ad_{H_{\nu'}}}}
	 e^{-i\tau_2 \ad_{H_\nu}}
	 \comm{H_\nu,H_\mu}\\
	 &=-i\sum_{\nu=\mu+1}^L \comm{H_\nu,H_\mu}\tau_1 -i  \sum_{\nu = \mu+1}^L \int_0^{\tau_1}d\tau_2 
	 \underbrace{
	 \left(\prod_{\nu' = \nu+1}^{L}{e^{-i\tau_1 \ad_{H_{\nu'}}}}
	 e^{-i\tau_2 \ad_{H_\nu}}
	 \comm{H_\nu,H_\mu} - \comm{H_\nu,H_\mu}\right)}_{\equiv G_{\mu,\nu}(\tau_1,\tau_2)} .
\end{align}
We note that $G(\tau_1) = \O{\tau_1+\tau_2}$ [Recall that $\O{}$ is the standard Bachmann-Landau big-$O$ notation.]
Therefore, we can rewrite it (using either \cite[]{childsTheoryTrotterError2020} or a direct differentiation) as
\begin{align} 
 	G_{\mu,\nu}(\tau_1,\tau_2) = -i \sum_{\nu' = \nu+1}^L \prod_{s = \nu'+1}^L e^{-i\tau_1\ad_{H_{s}}}
 	\int_0^{\tau_1} d\tau_3 e^{-i\tau_3\ad_{H_{\nu'}}}\comm{H_{\nu'},\comm{H_\nu,H_\mu}}\nonumber\\
 	-i\frac{\tau_2}{\tau_1}\prod_{\nu' = \nu+1}^L e^{-i\tau_1\ad_{H_{\nu'}}}
 	\int_0^{\tau_1} d\tau_3 e^{-i\tau_3\frac{\tau_2}{\tau_1}\ad_{H_{\nu}}}\comm{H_{\nu},\comm{H_\nu,H_\mu}}.
\end{align} 
Using the triangle inequality, we have
\begin{align} 
	\norm{G_{\mu,\nu}(\tau_1,\tau_2)} \leq \tau_1 \sum_{\nu'=\nu}^L \norm{\comm{H_{\nu'},\comm{H_\nu,H_\mu}}}.
\end{align}
Therefore, with $v_0 = \sum_{\mu = 1}^L\sum_{\nu = \mu+1}^L \comm{H_\nu,H_\mu}$, we have
\begin{align} 
	F(\tau_1) + iv_0 \tau_1
	&= \underbrace{-i \sum_{\mu=1}^L \sum_{\nu = \mu+1}^L
		  \int_0^{\tau_1} d\tau_2 G_{\mu,\nu}(\tau_1,\tau_2)}_{\equiv \mathcal{F}(\tau_1)}.\label{eq:expbound43204u392}
\end{align}
Using the bound on $\norm{G_{\mu,\nu}}$, we have
\begin{align} 
	 \norm{\mathcal{F}(\tau_1)} \leq \tau_1^2  \sum_{\mu=1}^L \sum_{\nu = \mu+1}^L \sum_{\nu' = \nu}^{L}\norm{\comm{H_{\nu'},\comm{H_\nu,H_\mu}}},
\end{align}
which implies
\begin{align} 
	\norm{F(\tau_1)} &\leq 
	\tau_1  \underbrace{\sum_{\mu = 1}^{L} \sum_{\nu = \mu+1}^L \norm{\comm{H_\nu,H_\mu}}}_{\equiv \alpha} 
	+{\tau_1^2} \underbrace{\sum_{\mu = 1}^{L} \sum_{\nu = \mu+1}^L\sum_{\nu' = \nu}^L \norm{\comm{H_{\nu'},\comm{H_\nu,H_\mu}}} }_{\equiv \beta}.
\end{align} 
In addition, combining \cref{eq:expbound43204u392} with \cref{eq:experrorapp}, we have
\begin{align} 
	S_\dt &=  
	\Texp{-i\int_0^\dt d\tau_1 \left[H - i v_0 \tau_1 +\mathcal{F}(\tau_1)\right] },\label{eq:experrorallterms}
\end{align}
with $v_0$ and $\mathcal{F}(\tau_1)$ given above.

Next, we rewrite the time-ordered exponential into a regular exponential using the Magnus expansion.

\begin{lemma}[Magnus expansion {\cite{blanesMagnusExpansionIts2009a,Moan2008,Arnal_2018}}]
	\label{lem:magnus}
	Let $\mathcal{A}(\tau)$ be a continuous operator-valued function defined for $0\leq\tau\leq t$ such that $\int_{0}^{t}\mathrm{d}\tau\norm{\mathcal A(\tau)}\leq 1$. Then, the equality 
	\begin{equation}
		\Texp{\int_{0}^{t}\mathrm{d}\tau\ \mathcal{A}(\tau)}
		=\exp\bigg(\sum_{j=1}^{\infty}\Omega_j(t)\bigg)
	\end{equation}
	holds with a convergent operator series $\sum_{j=1}^{\infty}\Omega_j(t)$, where
	\begin{align} 
		\Omega_j(t) = \frac{1}{j^2} \sum_{\sigma} (-1)^{d_b} \frac{1}{\binom{j-1}{d_b}} \int_0^t d\tau_{1}\dots \int_0^{d\tau_{j-1}} d\tau_j  \comm{\mathcal A(\tau_1),\dots\comm{\mathcal A(\tau_{j-1}),\mathcal A(\tau_j)}\dots},
	\end{align}
	with the sum being taken over all permutations $\sigma$ of $\{1,\dots,j\}$ and $d_b$ is the number of descents, i.e. pairs of consecutive numbers $\sigma_k,\sigma_{k+1}$ for $k = 1,\dots,j-1$ such that $\sigma_k>\sigma_{k+1}$, in the permutation $\sigma$. Furthermore, $\Omega_j(t)$ are all anti-Hermitian if $\mathcal{A}(\tau)$ is anti-Hermitian.
	It is worth noting that the first two $\Omega_j(t)$ are
    \begin{align} 
    	&\Omega_1(t)=\int_{0}^{t}\mathrm{d}\tau\ \mathcal{A}(\tau),\\
    	&\Omega_2(t)=\frac{1}{2}\int_{0}^{t}\mathrm{d}\tau_1\int_{0}^{\tau_1}\mathrm{d}\tau_2\comm{\mathcal{A}(\tau_1),\mathcal{A}(\tau_2)}.
    \end{align}
\end{lemma}

We now use \cref{lem:magnus} to rewrite \cref{eq:experrorallterms} with $\mathcal A(\tau) = -i[H + F(\tau)] = -i [H - iv_0 \tau_1 + \mathcal{F}(\tau_1)]$:
\begin{align} 
	S_\dt = \exp\left(\sum_{j = 1}^\infty \Omega_j(\dt)\right), 
\end{align}
where the first-order Magnus term is
\begin{align} 
	\Omega_1(\dt) &=  -i\int_0^\dt d\tau_1 \left[H - iv_0\tau_1+ \mathcal{F}(\tau_1) \right]
	= -i\dt \left[H - \frac{i}{2}v_0 \dt + \frac{1}{\dt}\int_0^\dt d\tau_1 \mathcal{F}(\tau_1)\right].
\end{align}
To bound the higher-order terms in the Magnus expansion, we first note that
\begin{align} 
	\normcomm{\mathcal A(\tau_1),\mathcal A(\tau_2)}
	&= \normcomm{
	H + F(\tau_1),
	H + F(\tau_2)
	}\\
	&\leq 
	2\bigg(\norm{H}\norm{F(\tau_2)}
	+\norm{F(\tau_1)}\norm H
	+\norm{F(\tau_1)}\norm{F(\tau_2)}\bigg)\\
	&
	\leq 2\bigg(
	\norm{H}\left(\alpha\tau_1 + \beta\tau_1^2\right)
	+\norm{H}\left(\alpha\tau_2 + \beta\tau_2^2\right)
	+\left(\alpha\tau_1 + \beta\tau_1^2\right)\left(\alpha\tau_2 + \beta\tau_2^2\right)
	\bigg)\\
	&\leq 2\left(
	2\norm{H}+\alpha\dt + \beta\dt^2\right)\left(\alpha\dt + \beta\dt^2\right)\\
	&\leq 4\left(
	\norm{H}+\alpha\dt + \beta\dt^2\right)\left(\alpha\dt + \beta\dt^2\right)
\end{align}
for all $\tau_1,\tau_2\leq \dt$.
Similarly, for higher-order nested commutators:
\begin{align} 
	\normcomm{\mathcal A({\tau_1}),\dots\comm{\mathcal A(\tau_{j-1}),\mathcal A(\tau_j)}\dots }
	&\leq 2^{j-2} \norm{\mathcal A({\tau_1})}\dots \norm{\mathcal A({\tau_{j-2}})} \norm{\comm{\mathcal A(\tau_{j-1}),\mathcal A(\tau_j)}}\\
	&\leq 2^j \left(
	\norm{H}+\alpha\dt + \beta\dt^2\right)^{j-1} \left(\alpha\dt + \beta\dt^2\right).
\end{align}
Using \cref{lem:magnus} and noting that there are $j!$ permutations for each $j$, we can crudely bound
\begin{align} 
	\norm{\Omega_j(\dt)} &\leq  \sum_{\sigma} \int_0^\dt d\tau_1 \int_0^{d\tau_{j-1}} d\tau_j \normcomm{\mathcal A({\tau_1}),\dots\comm{\mathcal A(\tau_{j-1}),\mathcal A(\tau_j)}\dots }\\
	&\leq j! \frac{\dt^j}{j!}2^j \left(
	\norm{H}+\alpha\dt + \beta\dt^2\right)^{j-1} \left(\alpha\dt + \beta\dt^2\right)\\
	&\leq (2\dt)^j \left(
	\norm{H}+\alpha\dt + \beta\dt^2\right)^{j-1} \left(\alpha\dt + \beta\dt^2\right)
\end{align}
for all $j\geq 2$.
Define
\begin{align} 
	\mathcal V(\dt) \equiv \frac{1}{\dt} \int_0^\dt d\tau_1 \mathcal{F}(\tau_1) +\frac{1}{\dt} \sum_{j=2}^\infty \Omega_j(\dt), 
\end{align}
we could write
\begin{align} 
	S_\dt = \exp\bigg[-i\dt \big(H  \underbrace{- \frac{i}{2}v_0 \dt + \mathcal V(\dt)}_{\equiv V}\big)\bigg]. 
\end{align}
It follows from the bounds on $\Omega_j$ above that
\begin{align} 
	\norm{\mathcal V(\dt)}
	&\leq \dt^2 \beta +  (\alpha+\beta \dt) 
	\sum_{j=2}^{\infty}(2\dt)^j\left(
	\norm{H}+\alpha\dt + \beta\dt^2
	\right)^{j-1}\\
	&\leq \dt^2 \beta + 4\dt^2 (\alpha+\beta \dt) \left(
	\norm{H}+\alpha\dt + \beta\dt^2
	\right)
	\sum_{j=0}^{\infty}(2\dt)^j\left(
	\norm{H}+\alpha\dt + \beta\dt^2
	\right)^{j}\\
	&\leq \dt^2 \beta + 8\dt^2 (\alpha+\beta \dt) \left(
	\norm{H}+\alpha\dt + \beta\dt^2
	\right),
\end{align}
where we have assumed $\dt\big(
	\norm{H}+\alpha\dt + \beta\dt^2
	\big)\leq 1/4$ so that the sum over $j$ in the second line converges.
	We note that this assumption also ensures that the Magnus expansion converges.
The bound states that $\norm{\mathcal V(\dt)}$ scales with $\dt$ as $\O{\dt^2}$.

Assuming $\beta \dt \leq \alpha$, $2\alpha \dt \leq \norm H$, and $8\dt \norm H\leq 1$, we get
\begin{align} 
	\norm{\mathcal V(\dt)} \leq \dt^2 \left(\beta + 32\alpha\norm H\right). 
\end{align}
This bound completes the proof of \cref{lem:Htilde}. Note that the constant prefactor of our bound may be further tightened by using a stronger version of \cref{lem:magnus}. Such an improvement may be especially useful for near-term implementations of quantum simulation, but a detailed discussion falls out of the scope of the current paper and will be left as a subject for future investigation.
\end{proof}
\section{Proof of \cref{lem:pf1finalbound}}\label{sec:thm1proof}

In this section, we provide more details on the proof of \cref{lem:pf1finalbound} for completeness.
Using the triangle inequality
\begin{align} 
	\epsilon & =  \norm{\prod_{k=1}^re^{-i (H+C_k V C_k) \dt} - e^{-iHt}}\nonumber\\
	&\leq   \norm{e^{-i\overbar H_\eff t} - e^{-iHt}}+ \norm{\prod_{k=1}^re^{-i (H+C_k V C_k) \dt} - e^{-i\overbar H_\eff t}}\\
	&\leq \norm{\overbar{V}} t+ \frac{2\xi\sqrt{m}(\norm{H}+\norm{V})\norm{V}t^2\log r}{r} \\
	&\leq \frac{t^2}{2r} \norm{\overbar {v_0}}+ \chi \frac{t^3}{r^2} + \frac{2\xi\sqrt{m}(\norm{H}+\frac{1}{2}\alpha \dt + \chi \dt^2)(\frac{1}{2}\alpha \dt + \chi \dt^2)t^2\log r}{r}.
\end{align}
Since $\chi = \beta + 32\alpha\norm H$, $\chi \dt = \beta \dt + 32\alpha \norm H\dt \leq 5\alpha$ (assuming $\beta \dt \leq \alpha$ and $8\norm H\dt \leq 1$).
Therefore, we could upper bound
\begin{align} 
	(\norm{H}+\frac{1}{2}\alpha \dt + \chi \dt^2)(\frac{1}{2}\alpha \dt + \chi \dt^2)
	< 6(\norm{H}+6\alpha \dt)\alpha \dt 
	\leq 24\norm{H}\alpha \frac tr, 
\end{align}
where we have also used the assumption that $2\alpha\dt\leq \norm H$.
Therefore, we have
\begin{align} 
	\epsilon \leq  \frac{t^2}{2r} \norm{\overbar {v_0}}+ \chi \frac{t^3}{r^2} + \underbrace{48\xi\sqrt{m}\alpha\norm H}_{\equiv \kappa}\frac{t^3\log r}{r^2}.
\end{align}
This completes the proof of \cref{lem:pf1finalbound}.

\end{widetext}
\end{document}